\documentclass[preprint,12pt]{elsarticle}
\makeatletter
\nopreprintlinetrue
\makeatother

\usepackage{amssymb}
\usepackage{amsmath}
\usepackage{mathrsfs}
\usepackage{amsthm}
\usepackage[all]{xy}
\usepackage{graphicx}
\usepackage{color}
\usepackage{booktabs}
\usepackage{paralist}
\usepackage{tikz}
\usetikzlibrary{automata, arrows.meta, positioning, shapes, arrows}
\usepackage{multirow}
\usepackage{pifont}
\newcommand{\cmark}{\ding{51} }%
\newcommand{\xmark}{\ding{55} }%

\theoremstyle{plain}
\newtheorem{theorem}{Theorem}[section]
\newtheorem{proposition}[theorem]{Proposition}
\newtheorem{lemma}[theorem]{Lemma}
\newtheorem{corollary}[theorem]{Corollary}

\theoremstyle{definition}
\newtheorem{definition}[theorem]{Definition}
\newtheorem{example}[theorem]{Example}

\theoremstyle{remark}
\newtheorem{remark}[theorem]{Remark}

\newcommand{\red}[1]{}
\newcommand{\camera}[1]{{#1}}
\newcommand{\mc}[1]{\mathcal{#1}}
\newcommand{\structure}[2]{{(\mc{#1}, \mc{#2})}}
\newcommand{\class}[1]{\mathscr{#1}}
\newcommand{\proved}[0]{ \red{(proved)} }
\newcommand{\checked}[0]{ \red{(checked)} }
\newcommand{\checkedC}[0]{ \red{(checked when condition add on H)} }

\newcommand{\hypograph}[2]{({#1}, ({#2}_\syntax{R})_{\syntax{R}\in\mc{R}})}
\newcommand{\quotient}[2]{{#1}/{#2}}
\newcommand{\hypset}[2]{\mc H({#1}, {#2})}
\newcommand{\property}[0]{Prop}
\newcommand{\propH}[0]{PropH}
\newcommand{\propertyH}[2]{\propH({#1}, {#2})}
\newcommand{\propG}[0]{PropG}
\newcommand{\propertyG}[2]{\propG({#1}, {#2})}
\newcommand{\hypgraphset}[2]{\mc{HG}({#1}, {#2})}
\newcommand{\formgraphset}[0]{\mc{FG}}
\newcommand{\syntax}[1]{\mathtt{#1}}
\newcommand{\sentenceSet}[0]{{\mc L}_S}
\newcommand{\targetsetX}[0]{\mathcal{X}_t} 
\newcommand{\lastsetXL}[0]{\mathcal{X}_l}
\newcommand{\targetX}[0]{X_t}
\newcommand{\targetH}[0]{H_t}
\newcommand{\targetG}[0]{G_t}
\newcommand{\targetsetG}[0]{\mathcal{G}_t}
\newcommand{\targetsetH}[0]{\mathcal{H}_t}
\newcommand{\varsetV}[0]{\mathcal{V}}
\newcommand{\sortsetS}[0]{\mathcal{S}}
\newcommand{\constsetC}[0]{\mathcal{C}}
\newcommand{\propsetP}[0]{\mathcal{P}}
\newcommand{\relsetR}[0]{\mathcal{R}}

\newcommand{\FOtermset}[0]{Term}
\newcommand{\cardinal}[1]{|#1|}
\newcommand{\fami}[3]{(#1)_{#2\in #3}}
\newcommand{\powerset}[1]{\wp(#1)}
\newcommand{\protocolTowards}[2]{{#1}[{#2}]}
\newcommand{\protocolname}[1]{\mathbf{#1}}
\newcommand{\protocolProp}[2]{{#1}({#2})}
\newcommand{\protocolSizeRestrict}[2]{{#1}_{#2}}
\newcommand{\UFBD}[0]{\protocolname{UFBD}}
\newcommand{\basicUFBD}[0]{\protocolname{Basic}}
\newcommand{\simpleUFBD}[0]{\protocolname{Simple}}
\newcommand{\simpleXbasicFUFBD}[0]{\protocolname{SimpleX}\textbf{-}\protocolname{BasicF}}
\newcommand{\basicXsimpleFUFBD}[0]{\protocolname{BasicX}\textbf{-}\protocolname{SimpleF}}

\date{}
\begin{document}

\begin{frontmatter}



\title{A Logical Formalisation of a Hypothesis in Weighted Abduction: towards User-Feedback Dialogues\tnoteref{label1}}
\tnotetext[label1]{A preliminary version of this article, excluding Sections \ref{subsection: Abstract User-Feedback Dialogue Protocols} \ref{subsection: Basic UFBD protocols} and \ref{section: other UFBD protocols}, was presented at ECSQARU2023 \cite{motoura2023logic}.}

\author{Shota Motoura\corref{cor1}}
\ead{motoura@nec.com}
\cortext[cor1]{Corresponding author}
\author{Ayako Hoshino}
\ead{ayako.hoshino@nec.com}

\author{Itaru Hosomi}
\ead{i-hosomi@nec.com}

\author{Kunihiko Sadamasa}
\ead{sadamasa@nec.com}

\address{NEC Corporation, 1753 Shimonumabe,\break Nakahara-ku, Kawasaki, Kanagawa 211-8666, Japan}

\begin{abstract}
%
Weighted abduction computes hypotheses that explain input observations.
A reasoner of weighted abduction first generates possible hypotheses and then selects the hypothesis that is the most plausible.
Since a reasoner employs parameters, called weights, that control its plausibility evaluation function, it can output the most plausible hypothesis according to a specific application using application-specific weights.
This versatility makes it applicable from plant operation to cybersecurity or discourse analysis.
However, the predetermined application-specific weights are not applicable to all cases of the application.
Hence, the hypothesis selected by the reasoner does not necessarily seem the most plausible to the user.
 In order to resolve this problem, this article proposes two types of user-feedback dialogue protocols, in which the user points out, either positively, negatively or neutrally, properties of the hypotheses presented by the reasoner, and the reasoner regenerates hypotheses that satisfy the user's feedback.
As it is required for user-feedback dialogue protocols, we then prove:
\begin{inparaenum}[(i)]
\item
our protocols necessarily terminate under certain reasonable conditions;
\item
they achieve hypotheses that have the same properties in common as fixed target hypotheses do in common if the user determines the positivity, negativity or neutrality of each pointed-out property based on whether the target hypotheses have that property. 
\end{inparaenum}
\end{abstract}




\end{frontmatter}



\section{Introduction}
\emph{Abduction} is inference to the best explanation:
given a set of observations, abduction generates a set, called a \emph{hypothesis}, of causes that accounts for the observations and selects the most plausible hypothesis by some criterion \cite{aliseda2017logic,paul1993approaches}. 
Several frameworks for abduction have been proposed, including plausibility criteria and reasoning processes (cf. \cite{paul1993approaches}).
Among others, \emph{weighted abduction} generates proofs, called \emph{hypothesis graphs}, of the observations in first-order logic that explain how the hypothesised causes lead to the observations.
It has parameters, called weights, that control the plausibility evaluation function and this enables the abductive reasoner to select the most plausible hypothesis graph for a specific application. This versatility makes it applicable to plant operation \cite{motoura2018translating}, cybersecurity \cite{motoura2021cooperative}, discourse analysis \cite{hobbs-etal-1988-interpretation,INOUE201433} or plan ascription \cite{appelt1992weighted}. As examples, let us see an application in plant operation \cite{motoura2018translating} and another in cybersecurity \cite{motoura2021cooperative}.
\begin{example}\label{example: MFM graph}
An abductive reasoner is used to compute an operation plan for a plant in \cite{motoura2018translating}. To illustrate this, let us consider a simple structure consisting of a faucet $F$ and a pipe $P$ coming out of it. Assume that the target state is the high flow in $P$ ($\syntax{Hold(P, High)}$). With $\syntax{Hold(P, High)}$ as an input observation, the reasoner outputs the hypothesis graph below (the hypergraph is simplified from one in  \cite{motoura2018translating}):

{\scriptsize
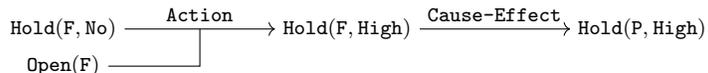
\begin{center}
\begin{tikzpicture}[node distance =0cm and 1cm]
\node [name=q00] {$\syntax{Hold(F, No)}$};
\node [below = of q00] (q01) {$\syntax{Open(F)}$};
\coordinate [right = of q00, label=above:$\syntax{Action}$] (q10) ;
\node [right = of q10] (q20) {$\syntax{Hold(F, High)}$};
\draw [-] (q01) -| (q10);
\draw [->] (q00) -- (q20);\
\coordinate [right = of q20, label=above:\texttt{Cause-Effect}] (q30) ;
\node [right = of q30] (q40) {$\syntax{Hold(P, High)}$};
\draw [->] (q20) -- (q40);
\end{tikzpicture}
\end{center}
}
\noindent This hypothesis graph expresses that opening $F$ ($\syntax{Open(F)}$) changed its flow from no flow ($\syntax{Hold(F, No)}$) to high flow ($\syntax{Hold(F, High)}$), which in turn caused the high flow in $P$ ($\syntax{Hold(P, High)}$). In \cite{motoura2018translating}, this hypothesis is seen as an operational procedure where $\syntax{Open(F)}$ leads to $\syntax{Hold(P, High)}$ from the current state, where $\syntax{Hold(F, No)}$ holds.\hfill$\dashv$
\end{example}

\begin{example}\label{example: cybersec}
In \cite{motoura2021cooperative}, abductive reasoning is applied to cybersecurity.
Below is an example of an output of its reasoner: (the nodes with * are observations)
{\scriptsize
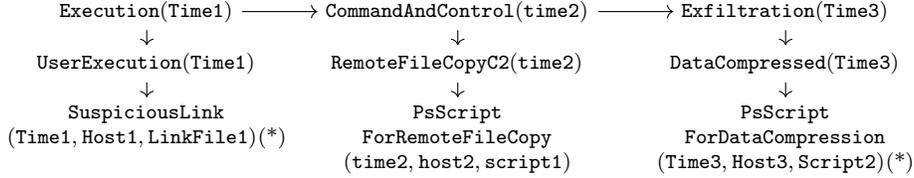
\begin{center}
\begin{tikzpicture}[node distance =0.2cm and 1cm]
\node [name=q00] {$\syntax{Execution(Time1)}$};
\node [right = of q00] (q10) {$\syntax{CommandAndControl(time2)}$};
\node [right = of q10] (q20) {$\syntax{Exfiltration(Time3)}$};
\node [below = of q00] (q01) {$\syntax{UserExecution(Time1)}$};
\node [below = of q10] (q11) {$\syntax{RemoteFileCopyC2(time2)}$};
\node [below = of q20] (q21) {$\syntax{DataCompressed(Time3)}$};
\node [below = of q01, align=center] (q02) {$\syntax{SuspiciousLink}$\\$\syntax{(Time1, Host1, LinkFile1)}$(*)};
\node [below = of q11, align=center] (q12) {$\syntax{PsScript}$\\$\syntax{ForRemoteFileCopy}$\\$\syntax{(time2, host2, script1)}$};
\node [below = of q21, align=center] (q22) {$\syntax{PsScript}$\\$\syntax{ForDataCompression}$\\$\syntax{(Time3, Host3, Script2)}$(*)};
\draw [->] (q00) -- (q10);
\draw [->] (q10) -- (q20);
\draw [->] (q00) -- (q01);
\draw [->] (q01) -- (q02);
\draw [->] (q10) -- (q11);
\draw [->] (q11) -- (q12);
\draw [->] (q20) -- (q21);
\draw [->] (q21) -- (q22);
\end{tikzpicture}
\end{center}
}
\noindent
This is in the form of a TTP (Tactics, Techniques and Procedures) framework (cf. \cite{hutchins2010intelligence,mitre_attack}).
The terms starting with a capital letter, such as $\syntax{Time1}$, are constants, and the terms starting with a lower-case letter, such as $\mathtt{host2}$, are variables, which are implicitly existentially quantified.
The vertical edges express end-means relation while the horizontal ones do the chronological order. $\hfill\dashv$

\end{example}

Although its weights enable an abductive reasoner to output the hypothesis that is the most plausible for a specific application, the predetermined application-specific weights are not applicable to all cases.
Hence, the hypothesis selected by a reasoner does not necessarily seem the most plausible to the user.
To address this issue, user-feedback functions are proposed in \cite{motoura2021cooperative}, by which the user can give feedback on nodes in the hypothesis graphs presented by the reasoner. However, these functions are \textit{ad-hoc}.
In particular, it is theoretically unclear whether feedback on nodes has sufficient expressivity so that the user can achieve the hypothesis graph that is the most plausible to him if he gives enough feedback.

This is partly due to the absence of a formal definition of a hypothesis graph in applications.
The concept of a hypothesis graph in applications is extended from the original concept of a hypothesis graph, i.e. a proof in first-order logic.
For example, hyperedge $\syntax{Action}$ in Example \ref{example: MFM graph}  is not implication since $\syntax{Hold(F, No)}$ contradicts $\syntax{Hold(F, High)}$ and the arguments of $\syntax{Hold}$ are many-sorted; precisely, the first argument is a component of the plant and the second is a state of the component.  
In Example \ref{example: cybersec}, the edges do not express implication but they do end-means relation or the chronological order, and the hypothesis graph contains many-sorted constants and variables, such as $\syntax{Time1}$ and $\syntax{host2}$.
\paragraph{Contributions of This Article}
To address these issues, we propose a formal definition of an extended hypothesis graph and two types of user-feedback dialogue protocols based on the definition.
More precisely, our contributions are as follows.
\begin{enumerate} 
\item
We introduce a variant of second-order logic whose language contains many-sorted constants and variables as well as second-order predicate symbols whose arguments are first-order literals.
We then define the concepts of a hypothesis and a hypothesis graph such as ones in Examples \ref{example: MFM graph} and \ref{example: cybersec}.
\item
Using our definitions of a hypothesis and a hypothesis graph, we propose two types, $\basicUFBD$ and $\simpleUFBD$, of user-feedback dialogue protocols, in which the user points out, either positively, negatively or neutrally, properties of the hypotheses/hypothesis graphs presented by the abductive reasoner, and the reasoner regenerates hypotheses/hypothesis graphs that satisfy the user's feedback. 
As it is required for user-feedback protocols, we prove:
\begin{inparaenum}[(i)]
\item
our protocols necessarily terminate under certain reasonable conditions, \emph{the halting property};
\item
they achieve hypotheses/hypothesis graphs that have the same properties in common that fixed target hypotheses/hypothesis graphs do in common, if the user determines the positivity, negativity and neutrality of each pointed-out property based on whether the target hypotheses/hypothesis graphs have that property, \emph{the convergence property}. 
\end{inparaenum}
The results are summarised in Table \ref{table; convergence and holting results}.
All feedback-dialogue protocols in the table satisfy the halting property under certain reasonable conditions and that
the two protocols of $\basicUFBD$ satisfy the convergence property as well, whilst the two protocols $\simpleUFBD$ satisfy when the target is a singleton
\end{enumerate}

\begin{table}
\centering
{\scriptsize
\begin{tabular}{l cccc}
\multicolumn{1}{c}{\multirow{2}{*}{Protocol under Conditions}} &\multicolumn{2}{c}{Convergence Property}    &\multicolumn{2}{c}{Halting Property}  \\
\cmidrule(lr){2-3}  \cmidrule(lr){4-5} 
& set& singleton & set & singleton \\ \hline
 $\basicUFBD$ on hypotheses & \cmark Theorem \ref{theorem: convergence} &  \cmark  & \cmark Theorem \ref{theorem: halting property of syntax-based dialogue protocol} & \cmark\\
$\basicUFBD$ on hypothesis graphs & \cmark Theorem \ref{the last set with Condition 2e} & \cmark &  \cmark Theorem \ref{DecidabilityOfGraphUFBD} &  \cmark\\ \hline
$\simpleUFBD$ on hypotheses &\xmark Example \ref{example: simple UFBD not convergence propH to set}& \cmark Theorem \ref{theorem: simple convergence towards single target} &   \cmark Corollary \ref{corollary: simple hypotheses halting}& \cmark \\
$\simpleUFBD$ on hypothesis graphs&\xmark Example \ref{example: target set in simple UFBD}& \cmark Theorem \ref{theorem: pointwise convergence of simple n-bounded on graph}& \cmark Corollary \ref{corollary: simple graph n-bounded halting} & \cmark \\
\end{tabular}
}
\caption{Satisfaction relation under certain reasonable conditions between user-feedback dialogue protocols proposed in this article and the halting and the convergence properties.
$\basicUFBD$ and $\simpleUFBD$ are types of user-feedback dialogue protocols.
The labels `set' and `singleton' indicate that the target is a set of hypotheses/hypothesis graphs or a singleton.
Theorem, Corollary or Example in a cell is one in this article that supports the result of the cell.
The result of a cell containing \cmark or \xmark only follows from the fact that a singleton is a special case of a set.}
\label{table; convergence and holting results}
\end{table}
\paragraph{Organisations}
In Section \ref{section: logical framework}, we propose formal definitions of an extended hypothesis and an extended hypothesis graph.
We then propose two types of user-feedback dialogue protocols on them and prove that they satisfy the halting and the convergence properties, mentioned above, in Section \ref{section: feedback protocols}.

\section{Related Work}
Several definitions of a hypothesis have been proposed on the basis of first-order logic \cite[Section 2.2]{paul1993approaches}.
In particular, a hypothesis graph in weighted abduction is a proof in first-order logic consisting of literals \cite{stickel1991prolog}.
Thus, the edges between literals of a hypothesis graph express implication.
Extending this definition, we propose a definition of a hypothesis graph that is allowed to contain many-sorted constants and variables and arbitrarily labelled hyperedges between literals.

Dialogue protocols have also been proposed in which a system and a user cooperate in seeking for a proof or a plausible hypothesis.
An inquiry dialogue \cite{black2009inquiry} is one of such protocols. Its goal is to construct a proof of a given query such as `he is guilty' in propositional logic.
Using this protocol, the user can obtain all minimal proofs and find the one that seems most plausible to him.
However, this cannot be applied to extended hypotheses/hypothesis graphs, because there can be an infinite number of possible extended hypotheses/hypothesis graphs (see also Examples \ref{example: semantics-based protocol not terminate} and \ref{example: infinite syntax-based dialogue}, \textit{infra}).
The work \cite{motoura2021cooperative} proposes user-feedback functions on nodes in extended hypothesis graphs presented by the abductive reasoner.
However, these functions lack theoretical and empirical supports. 
In contrast, we theoretically prove that our protocols enjoy the halting and the convergence properties.

\section{Hypotheses and Hypothesis Graphs}\label{section: logical framework}
In this section, we first introduce a variant of many-sorted second-order logic and then define, in the logic, the concepts of a hypothesis and a hypothesis graph such as ones in Examples \ref{example: MFM graph} and \ref{example: cybersec}.
Throughout this section, we take the hypergraph in Example \ref{example: MFM graph} as our running example.

\subsection{Language}
We first define the language, which is second-order in the sense that the arguments of a second-order predicate symbol are first-order literals.
\begin{definition}
An \emph{alphabet} \camera{is a tuple $\Sigma=(\mathcal{S}, \mathcal{C}, \mathcal{P}, \mathcal{R},\mathcal{V}_1,\mathcal{V}_2)$ such that:
\begin{enumerate}
\item
$\mathcal{S}$ is a finite set of \emph{sorts} $\sigma$;
\item
$\mathcal{C}$ is an $\mathcal{S}$-indexed family $(\mathcal{C}_\sigma)_{\sigma\in\mathcal{S}}$ of finite sets $\mathcal{C}_\sigma$ of \emph{first-order constant symbols} $\syntax{c}:\sigma$ of sort $\sigma$;
\item
$\mc P$ is a non-empty finite set of \emph{first-order predicate symbols} $\syntax{p}: (\sigma_1,\ldots,\sigma_n)$, where $(\sigma_1,\ldots,\sigma_n)$ indicates the sorts of its arguments;
\item
$\mathcal{R}$ is a finite set of \emph{second-order predicate symbols} $\syntax{R}$;
\item
$\mathcal{V}_1$ is an $\mathcal{S}$-indexed family $(\mc{V}_\sigma)_{\sigma \in \mc S}$ of sets $\mc{V}_\sigma$ of \emph{first-order variables} $\syntax{x}:\sigma$ of sort $\sigma$;
\item
$\mc V_2$ is a set of \emph{second-order variables} $\syntax{X}$.
\end{enumerate}
We also use the first-order equality symbol ${=}: (\sigma,\sigma)$ for each $\sigma\in\mc S$ and the second-order one $=$ as logical symbols, that is, these symbols are interpreted in the standard manner.}
We often suppress the sorts of symbols and variables
if they are clear from the context.
\hfill $\dashv$
\end{definition}
\noindent
Note that an alphabet does not contain function symbols, following the original definition of weighted abduction \cite{stickel1991prolog}.
In this article, we assume that $\Sigma=(\mathcal{S}, \mathcal{C},\mathcal{P}, \mathcal{R},\mathcal{V}_1,\mathcal{V}_2)$ ranges over  alphabets except in examples.
\begin{example}
An alphabet $\Sigma=(\mathcal{S}, \mathcal{C},\mathcal{P}, \mathcal{R},\mathcal{V}_1,\mathcal{V}_2)$ for Example \ref{example: MFM graph} is one such that:
\begin{gather*}
\mc S= \{\mathsf{comp},\mathsf{state}\},\quad \mc C_\mathsf{comp}=\{\syntax{F}, \syntax{P}\}, \quad \mc C_\mathsf{state}=\{\syntax{High}, \syntax{No}\},\\
\mc P = \{\syntax{Hold}: (\mathsf {comp}, \mathsf{state}),\quad \syntax{Open}: (\mathsf {comp})\},\quad \mc R = \{\syntax{Action}, \syntax{Cause}\texttt{-}\syntax{Effect}\}.
\end{gather*}
 Here, $\mathsf{comp}$ means the sort of components. \hfill$\dashv$
\end{example}

\begin{definition}[Language]The language for $\Sigma=(\mathcal{S}, \mathcal{C},\mathcal{P}, \mathcal{R},\mathcal{V}_1,\mathcal{V}_2)$ is defined as follows.
\begin{enumerate}
\item
The \emph{first-order terms} of sort $\syntax{\sigma}\in \sortsetS$ are defined to be the first-order variables $\syntax{x:\sigma}\in\varsetV_\sigma$ and the constant symbols $\syntax{c:\sigma}\in\constsetC_\sigma$.
The set of first-order terms of $\sigma\in\sortsetS$ is denoted by $Term_\sigma$.
\item
The \emph{first-order atomic formulae} are the equalities $\syntax{t:\sigma={s}:\sigma}$ and $\syntax{p(t}_1 \syntax{:\sigma}_1\syntax{\ldots,t}_n\syntax{:\sigma}_n\syntax{)}$,
 where $\syntax{p: (\sigma}_1,\ldots,\syntax{\sigma}_n)$ ranges over $\mc{P}$, $\syntax{t:\sigma}$ and $\syntax{s:\sigma}$ over $Term_\sigma$ and $\syntax{t}_i\syntax{:\sigma}_i$ over $Term_{\sigma_i}$ for any $i=1,\ldots,n$.

\item
The \emph{ first-order literals} are atomic formulae $\syntax{A}$ and the negations $\syntax{\neg A}$ of atomic formulae.
The set of first-order literals is denoted by $L$.

\item
The {\it second-order terms} are first-order literals in $L$ and second-order variables in $\varsetV_2$.
We denote the set of second-order terms by  ${\mc T}_2$.

\item
The {\it formulae} are defined by the rule:
\[\syntax{\Phi}::=\syntax{A} \mid \syntax{X} \mid \syntax{T}_1\syntax{ = T}_2 \mid \syntax{R(T}_1,\ldots,\syntax{T}_n\syntax{)} \mid \syntax{\neg \Phi} \mid \syntax{\Phi\lor\Phi} \mid \syntax{(\exists x:\sigma) \Phi} \mid \syntax{(\exists X) \Phi}.\]
where $\syntax{A}$ ranges over the set of first-order atomic formulae, $\syntax{X}$ over $\mc V_2$ and $\syntax{T}_i$ over ${\mc T}_2$  for any $n=1,\ldots,n$, $\syntax{R}$ over $\relsetR$ with $n$ its arity and $\syntax{x:\sigma}$ over the set of first-order variables with $\sigma$ its sort.
We denote the set of formulae by $\mc L$.\footnote{Following the original definition of weighted abduction, we restrict the arguments of second-order predicates to first-order literals, which is used to prove Theorems \ref{theorem: halting property of syntax-based dialogue protocol} and \ref{DecidabilityOfGraphUFBD}.}
\end{enumerate}
We write a term and a formula defined above in $\syntax{typewriter}$ $\syntax{font}$.
\hfill $\dashv$
\end{definition}
\noindent Other logical connectives $\top$, $\bot$, $\land$, $\rightarrow$, $\leftrightarrow$, $\syntax{\forall x: \sigma}$ and $\syntax{\forall X}$ are defined as usual.
The concepts of \emph{freely occurring of a first-order/second-order variable} are also defined as usual.
A \emph{sentence} is a formula in $\mc L$ that does not contain any freely occurring variables.
We denote the set of sentences by $\sentenceSet$.

\subsection{Semantics}
A \emph{structure} has two components to interpret a formula of each order.
In particular, one component interprets first-order literals and the other does second-order predicate symbols.

\begin{definition}
A \emph{structure for $\Sigma$} is a pair $(\mc M, \mc I)$ consisting as follows.
\begin{itemize}
\item
The first component $\mc M=((M_\sigma)_{\sigma\in \mathcal{S}}, (C_\sigma)_{\sigma\in\mathcal{S}}, (\syntax{p}^\mc M)_{\syntax{p}\in \mc P})$ is a first-order structure consisting of:
\begin{itemize}
\item
a non-empty set $M_\sigma$ for each $\sigma\in\mc S$, called \emph{the domain of} $\sigma$;
\item
a function $C_\sigma:{\mc C}_\sigma\to M_\sigma$ for each $\sigma \in \mc S$; and
\item
a subset $\syntax{p}^\mc M\subseteq M_{\sigma_1}\times \cdots\times M_{\sigma_n}$ for each $\syntax{p}: (\sigma_1,\ldots,\sigma_n)\in\mc P$.
\end{itemize}
\item
The second component $\mc I=(I, (\syntax{R}^\mc I)_{\syntax{R}\in\mathcal{R}})$ is a pair of:
\begin{itemize}
\item
the set $I= \{*\syntax{p}(e_1,\ldots, e_n)\mid *\in \{\epsilon, \neg\}, \syntax{p}: (\sigma_1,\ldots,\sigma_n) \in \mc P \text{ and } e_i\in M_{\sigma_i}\text{ for }i=1,\ldots, n\}$, where $\epsilon$ is the null string; and
\item
an $n$-ary relation $\syntax{R}^\mc I$ on $I$ for each $\syntax{R}\in \mc R$, where $n$ is the arity of $\syntax{R}$.
\end{itemize}
\end{itemize}
We often write $(\syntax{c}:\sigma)^{\mc M}$ to mean $C_\sigma(\syntax{c}:\sigma)$. \hfill $\dashv$
\end{definition}
\begin{example}\label{example: small structure}
The hypergraph in Example \ref{example: MFM graph} determines the structure $(\mc M, \mc I)$ such that: 
\begin{itemize}
\item 
the first component $\mc M$ consists of  
\begin{itemize}
\item
$M_{\sf comp} = \{F, P\}$, $C_{\sf comp} = \{(\syntax{F}, F), (\syntax{P}, P)\}$, 
\item
$M_{\sf state} =  \{High, No\}$, $C_{\sf state} = \{(\syntax{High}, High), (\syntax{No}, No)\}$,
\item
$\syntax{Hold}^{\mc M}=\{(F, No), (F, High), (P, High)\}$,
 $\syntax{Open}^{\mc M}=\{F\}$; and
 \end{itemize}

 \item
 the second component $\mc I$ consists of 
  \begin{itemize}
  \item
 $I=\{*\syntax{Hold}(o, s) \mid *\in \{\epsilon, \neg\}, o\in M_{\sf comp}\text{ and } s\in M_{\sf state}\}\cup\{*\syntax{Open}(o)\mid *\in \{\epsilon, \neg\}\text{ and }  o\in M_{\sf comp} \}$,
 \item
$\syntax{Action}^{\mc I}=\{(\syntax{Hold}(F, No), \syntax{Open}(F), \syntax{Hold}(F, High)\}$, and
 \item
$\syntax{Cause}\texttt{-}\syntax{Effect}^{\mc I}=\{(\syntax{Hold}(F, High), \syntax{Hold}(P, High)\}$.\hfill $\dashv$
\end{itemize}
\end{itemize}
\end{example}

For a structure $(\mc M, \mc I)$ with $\mc M=((M_\sigma)_{\sigma\in \mathcal{S}}, (C_\sigma)_{\sigma\in\mathcal{S}}, (\syntax{p}^\mc M)_{\syntax{p}\in \mc P})$ and $\mc I=(I, (\syntax{R}^\mc I)_{\syntax{R}\in\mathcal{R}})$,
a \emph{first-order assignment} is an $\mc S$-indexed family $\mu_1=(\mu_\sigma)_{\sigma\in\mc S}$ of mappings $\mu_\sigma:\mathcal{V}_\sigma\rightarrow M_\sigma$ and a \emph{second-order assignment} is a mapping $\mu_2:\mc V_2\rightarrow I$.
The interpretation of terms is given as follows:
\begin{definition}
Let $\structure{M}{I}$ be a structure and $\mu_1=(\mu_\sigma)_{\sigma\in\mc S}$ and $\mu_2$ a first-order assignment and a second-order one for $\structure{M}{I}$, respectively.
\begin{enumerate}
\item
A first-order term is interpreted as usual:
\begin{itemize}
\item
$(\syntax{c}: \sigma)^\mc M[\mu_1]=(\syntax{c}: \sigma)^\mc M$ for any $\syntax{c}:\sigma\in \mc C_\sigma$;
\item
$(\syntax{x}:\sigma)^\mc M[\mu_1]=\mu_\sigma(\syntax{x}:\sigma)$ for any $\syntax{x}:\sigma\in\mc{V}_\sigma$.
\end{itemize}
\item
A second-order term  is interpreted as follows:
\begin{itemize}
\item  for any $\syntax{p}(\syntax{t}_1,\ldots, \syntax{t}_n)\in L$ (respectively, $\syntax{\neg p}(\syntax{t}_1,\ldots, \syntax{t}_n)\in L$),
\begin{gather*}
\syntax{p}(\syntax{t}_1,\ldots, \syntax{t}_n)^{\structure{M}{I}}[\mu_1, \mu_2]=\syntax{p}(\syntax{t}_1^\mc M[\mu_1], \ldots, \syntax{t}_n^\mc M[\mu_1])\\
\text{ (respectively, } \syntax{\neg p}(\syntax{t}_1,\ldots, \syntax{t}_n)^{\structure{M}{I}}[\mu_1, \mu_2]=\syntax{\neg p}(\syntax{t}_1^\mc M[\mu_1], \ldots, \syntax{t}_n^\mc M[\mu_1])\text{)};
\end{gather*}
\item
$\syntax{X}^{\structure{M}{I}}[\mu_1,\mu_2]=\mu_2(\syntax{X})$ for any $\syntax{X} \in\mc{V}_2$. \hfill $\dashv$
\end{itemize}
\end{enumerate}
\end{definition}
\noindent The interpretation of a formula in $\mc L$ is defined as below:
\begin{definition}[Interpretation]
Let $\structure{M}{I}$ be a structure such that  $\mc M=((M_\sigma)_{\sigma\in \mathcal{S}}, (C_\sigma)_{\sigma\in\mathcal{S}}, (\syntax{p}^\mc M)_{\syntax{p}\in \mc P})$ and $\mc I=(I, (\syntax{R}^\mc I)_{\syntax{R}\in\mathcal{R}})$, and $\mu_1$ and $\mu_2$ be a first-order assignment and a second-order one for $\structure{M}{I}$, respectively.
For any formula $\syntax{\Phi}\in \mc L$, the statement that \emph{$\syntax{\Phi}$ is satisfied in $\mc M$ by $\mu_1$ and $\mu_2$} (notation: $\structure{M}{I}\models\syntax{\Phi}[\mu_1, \mu_2]$) is inductively defined as follows.
\camera{
\begin{enumerate}
\item
The Boolean connectives are interpreted as usual:
\end{enumerate}
\[\begin{array}{l}
\structure{M}{I} \models  \syntax{\neg \Phi}[\mu_1, \mu_2] \Leftrightarrow \structure{M}{I} \not\models \syntax{\Phi}[\mu_1, \mu_2];\\
\structure{M}{I} \models \syntax{\Phi\lor\Psi} [\mu_1, \mu_2] \Leftrightarrow \structure{M}{I} \models \syntax{\Phi}[\mu_1, \mu_2] \text{ and } \structure{M}{I} \models \syntax{\Psi}[\mu_1, \mu_2].
\end{array}
\]
\begin{enumerate}\setcounter{enumi}{1}
\item
The first-order symbols and quantifiers are also interpreted as usual: 
\end{enumerate}
\[\begin{array}{l}
\structure{M}{I} \models \syntax{t}_1=\syntax{t}_2[\mu_1, \mu_2] \Leftrightarrow {\syntax{t}_1}^{\mc M}[\mu_1] = {\syntax{t}_2}^{\mc M}[\mu_1];\\
\structure{M}{I} \models \syntax{p}(\syntax{t}_1,\ldots,\syntax{t}_n)[\mu_1, \mu_2] \Leftrightarrow (\syntax{t}_1^{\mc M}[\mu_1],\ldots,\syntax{t}_n^{\mc M}[\mu_1])\in \syntax{p}^\mc M;\\
\structure{M}{I} \models (\exists \syntax{x}:\sigma)\syntax{\Phi}[\mu_1,\mu_2] \Leftrightarrow \structure{M}{I}\models \syntax{\Phi}[\mu_1(e/(\syntax{x}:\sigma)),\mu_2]\text{ for some }e\in M_\sigma.
\end{array}
\]
\begin{enumerate}\setcounter{enumi}{2}
\item
The second-order variables, symbols and quantifiers are interpreted by the following clauses:
\end{enumerate}
\[
\begin{array}{lcl}
\structure{M}{I} \models \syntax{X}[\mu_1, \mu_2] & \Leftrightarrow &
\begin{cases}
(e_1,\ldots, e_n)\in \syntax{p}^\mc M & (\mu_2(\syntax{X})=\syntax{p}(e_1,\ldots,e_n)) \\
(e_1,\ldots, e_n)\notin \syntax{p}^\mc M & (\mu_2(\syntax{X})=\neg \syntax{p}(e_1,\ldots,e_n) )
\end{cases}
\end{array}
\]
\[
\begin{array}{lcl}
\structure{M}{I} \models \syntax{T}_1 = \syntax{T}_2 [\mu_1,\mu_2] & \Leftrightarrow & \syntax{T}_1^\structure{M}{I}[\mu_1,\mu_2] = \syntax{T}_2^\structure{M}{I}[\mu_1,\mu_2] \\
\structure{M}{I} \models \syntax{R(T}_1,\ldots,\syntax{T}_n)[\mu_1,\mu_2] & \Leftrightarrow & (\syntax{T}_1^\structure{M}{I}[\mu_1,\mu_2],\ldots, \syntax{T}_n^\structure{M}{I}[\mu_1, \mu_2])\in \syntax{R}^\mc{I} \\
\structure{M}{I} \models (\exists \syntax{X})\syntax{\Phi}[\mu_1,\mu_2] & \Leftrightarrow & \structure{M}{I}\models \syntax{\Phi}[\mu_1,\mu_2(l/\syntax{X})]\text{ for some } l \in I.
\end{array}
\]

}\noindent Here, the assignment $\mu_1(e/(\syntax{x}:\sigma))$ (respectively, $\mu_2(l/\syntax{X})$) is the same as $\mu_1$ (respectively, $\mu_2$) except that it assigns $e$ to $\syntax{x}:\sigma$ (respectively, $l$ to $\syntax{X}$). \hfill $\dashv$
\end{definition}
\noindent
We say that a formula $\syntax{\Phi}$ \emph{is valid in a structure} $\structure{M}{I}$
if $\structure{M}{I} \models\syntax{\Phi}[\mu_1, \mu_2]$
for any assignments $\mu_1$ and $\mu_2$ and that \emph{$\syntax{\Phi}$ is valid in a class $\class{C}$ of structures}
if $\syntax{\Phi}$ is valid in any structure in $\class{C}$.
For sets $\Gamma$ and $\Delta$ of formulae in $\mc L$ and a class $\class C$ of structures,
we say that $\Gamma$ \emph{logically implies} $\Delta$ \emph{in} $\class{C}$ (notation: $\Gamma\models_\class{C} \Delta$) if any pair of a structure $\structure{M}{I}$ in $\class C$ and an assignment $[\mu_1, \mu_2]$ satisfying any formula in $\Gamma$ also satisfies any formula in $\Delta$.
We also say that $\Gamma$ and $\Delta$ \emph{are logically equivalent in} $\class{C}$ and write $\Gamma\leftrightarrow_{\class C}\Delta$ if both $\Gamma\models_\class{C} \Delta$ and $\Delta\models_\class{C} \Gamma$ hold. \\

In the rest of this article, when we consider a class $\class C$ of structures, we assume that we have $M_\sigma=M'_\sigma$ and $(\syntax{c}:\sigma)^{\mc M} = (\syntax{c}: \sigma)^{\mc M'}$ for any $(\mc M, \mc I), (\mc M', \mc I')\in \class C$ with  $\mc M=((M_\sigma)_{\sigma\in \mathcal{S}}, (C_\sigma)_{\sigma\in\mathcal{S}}, (\syntax{p}^\mc M)_{\syntax{p}\in \mc P})$ and  $\mc M'=((M'_\sigma)_{\sigma\in \mathcal{S}}, (C'_\sigma)_{\sigma\in\mathcal{S}}, (\syntax{p}^\mc {M'})_{\syntax{p}\in \mc P})$,  $\sigma\in\mc S$ and $\syntax{c}:\sigma \in \mc C_\sigma$.
\begin{example}\label{example: rules}
We can use a \emph{theory} $T$, i.e. a set $T$ of sentences, to restrict we consider the structures to the class $\class{C}(T)$ of those validating all sentences in $T$.
For example, in the settings in Example \ref{example: MFM graph}, suppose that its theory $T$ contains:
\[\syntax{\Phi}:=(\forall \syntax{XYZ})[\syntax{Action}(\syntax{X, Y, Z}) \rightarrow\bigwedge \left\{
\begin{array}{l}
\syntax{X}= \syntax{Hold}(\syntax{F}, \syntax{No}),\\
\syntax{Y}=\syntax{Open}(\syntax{F}), \\
\syntax{Z}=\syntax{Hold}(\syntax{F}, \syntax{High}),\\
\syntax{X},  \syntax{Y}, \syntax{Z}
\end{array}
\right\}].
\]
Then, any $\structure{M}{I}\in \class C(T)$ satisfies:
\begin{enumerate}
\item
$\syntax{Action}^{\mc I}\subseteq\{(\syntax{Hold}(F, No), \syntax{Open}(F), \syntax{Hold}(F, High))\}$ and 
\item
$\syntax{Hold}^{\mc M}\supseteq\{(F, No), (F, High)\}$ and $\syntax{Open}^{\mc M}\supseteq\{F\}$ if  $\syntax{Action}^{\mc I}\neq \emptyset$.
\end{enumerate}
These mean that,
if $\syntax{Action}$ is the case for some arguments, they must be $\syntax{Hold}(F, No)$,  $\syntax{Open}(F)$ and $\syntax{Hold}(F, High)$ and each argument must be true.
\hfill $\dashv$
\end{example}

\subsection{Definitions of a Hypothesis and a Hypothesis Graph}
We define the concepts of hypothesis and hypothesis graph on the basis of the syntax and semantics defined in the previous section.
The concept of hypothesis is defined as follows:
\begin{definition}[Hypotheses]\label{def: hypothesis}
Let $O$ be a set of sentences and $\class{C}$ a class of structures.
A \emph{hypothesis $H$ for $O$ in $\class{C}$} is a finite set of sentences
that satisfies $H \cup O\not\models_{\class{C}}\bot$.
We denote the set of hypotheses for $O$ in $\class C$ by $\hypset O {\class C}$.
\end{definition}
\noindent

\begin{remark}
According to \cite[Section 4]{paul1993approaches},
the minimum requirements for a set $H$ of sentences to be an ordinary first-order hypothesis for a set $O$ of observations under a theory $T$ are that $H$ implies $O$ under $T$ and that $H$ in conjunction with $O$ is consistent with $T$.
Since an extended  hypothesis graph does not necessarily implies the observations as in Example \ref{example: cybersec}, 
we adopt only the semantical counterpart $H \cup O\not\models_{\class{C}}\bot$ of the latter condition.
However, all results below in this article, except for Examples \ref{example: hyp graph}, \ref{example: UFBD}, \ref{example: semantics-based protocol not terminate}, \ref{example: infinite syntax-based dialogue}, \ref{example: target set in simple UFBD},
\ref{example: simple UFBD not convergence propH to set},
\ref{example: no-halting property Simple propH}, \ref{example: infinite syntax-based dialogue simple} and \ref{example: SimpleX-BasicF PropG non-convergence}, hold even if any condition is imposed on $H$ in Definition \ref{def: hypothesis}.\proved\checked
\hfill $\dashv$
\end{remark}

\camera{A hypothesis graph is a special case of a formula graph, defined below:}
\begin{definition}\label{definition: formula graph}
A \emph{formula graph} is a pair $\hypograph{V}{E}$ such that:
\begin{enumerate}
\item
The first component $V$ is a finite set of first-order literals;
\item
The second one $E_\syntax{R}$ for each $\syntax{R}\in \mathcal{R}$ is a set of tuples of first-order literals in $V$ whose length is equal to the arity of $\syntax{R}$.
\end{enumerate}
We denote the set of all formula graphs by $\formgraphset$. \hfill $\dashv$
\end{definition}
\noindent For a formula graph $G=\hypograph{V}{E}$, we call the cardinality $|V|$ of $V$ \emph{the order of} $G$ and write $|G|$.
We use the notation $Var_\sigma(G)$ for $\sigma\in\sortsetS$ to mean the set of variables of sort $\sigma$ occurring in some formula in $V$.

Due to its finiteness,
a formula graph $G=\hypograph{V}{E}$ can naturally be translated into the formula:
\[Form(G):=\bigwedge V \land \bigwedge \{\syntax{R}(\syntax{l}_1, \ldots, \syntax{l}_n) \mid \syntax{R} \in \mc{R}\text{ and } (\syntax{l}_1, \ldots, \syntax{l}_n)\in {E_\syntax{R}}\}.\]
We also denote the existential closure of $Form(G)$ by $Sent(G)$.
Using this translation, the concept of a hypothesis graph is defined:
\begin{definition}[Hypothesis Graphs]
A formula graph $G$ is a \emph{hypothesis graph for a set $O$ of sentences in a class $\class C$ of structures} if $\{Sent(G)\}$ is a hypothesis for $O$ in $\class C$.
We denote, by $\hypgraphset O {\class C}$, the set of hypothesis graphs for $O$ in $\class C$. \hfill $\dashv$
\end{definition}
\begin{example}\label{example: hyp graph}
Consider the hypergraph in Example \ref{example: MFM graph}, here referred to as $G_0$, the structure $\structure{M}{I}$ and the sentence $\syntax{\Phi}$ in Examples \ref{example: small structure} and \ref{example: rules}, respectively.
Then, $\structure{M}{I}$ validates $Sent(G_0)\land \syntax{Hold(P, High)}$ and thus we have $\{Sent(G_0)\} \cup \{\syntax{Hold(P, High)} \}\not\models_{\class{C}(\{\syntax{\Phi}\})}\bot$ since $\structure{M}{I}$ is in $\class{C}(\{\syntax{\Phi}\})$.
Hence, $\{Sent(G_0)\}$ and $G_0$ are hypothesis and a hypothesis graph for $\{\syntax{Hold(P, High)}\}$ in $\class C(\{\syntax{\Phi}\})$, respectively. \hfill$\dashv$
\end{example}

\subsection{Homomorphisms and Isomorphisms between Formula Graphs}
In this section, we define the concepts of homomorphism and an isomorphism from a formula graph to another and prove several lemmata and propositions, which we shall use in the rest of the article.

The concept of a homomorphism is defined using that of \emph{substitution}:
\begin{definition}
Let $\sigma$ be a sort in $\mc S$.
A \emph{substitution of $\sigma$} is a mapping $\alpha_\sigma: Term_\sigma\rightarrow Term_\sigma$ that satisfies $\alpha_\sigma(\syntax{c}:\sigma)=\syntax{c}:\sigma$ for any $\syntax{c}:\sigma\in {\mc C}_{\sigma}$. \hfill$\dashv$
\end{definition}
\begin{definition}
Let $G=\hypograph{V}{E}$ and $G'=\hypograph{V'}{E'}$ be formula graphs.
By a \emph{homomorphism} $(h, \alpha)$ from $G$ to $G'$ (notation: $(h, \alpha): G\rightarrow G'$), we mean a pair of
a mapping $h:V\rightarrow V'$ and
an $\mc S$-indexed family $\alpha=(\alpha_\sigma)_{\sigma \in \mc S}$ of substitutions $\alpha_\sigma$ of $\sigma\in \mc S$
such that:
\begin{enumerate}
\item
for any $\syntax{p}(\syntax{t}_1,\ldots, \syntax{t}_n)\in V$ (respectively, $\syntax{\neg p}(\syntax{t}_1,\ldots, \syntax{t}_n)\in V$),
\begin{gather*}
h(\syntax{p}(\syntax{t}_1, \ldots, \syntax{t}_n))=\syntax{p(\alpha}_{\sigma_1}\syntax{(t}_1), \ldots, \alpha_{\sigma_n}(\syntax{t}_n))\\
\text{(respectively, }h(\syntax{\neg p}(\syntax{t}_1, \ldots, \syntax{t}_n))= \syntax{\neg p(\alpha}_{\sigma_1}\syntax{(t}_1), \ldots, \alpha_{\sigma_n}(\syntax{t}_n)) \text{)}
\end{gather*}
holds, where $(\sigma_1,\ldots, \sigma_n)$ is the sorts of the arguments of $\syntax p$; and
\item
$(\syntax{l}_1,\ldots, \syntax{l}_m)\in E_\syntax{R}$ implies $ (h(\syntax{l}_1),\ldots, h(\syntax{l}_m))\in E'_\syntax{R}$ for any $\syntax{R} \in \mathcal{R}$ and $\syntax{l}_1,\ldots, \syntax{l}_m \in V$, where $m$ is the arity of $\syntax{R}$. 
\end{enumerate}
\hfill $\dashv$
\end{definition}
\noindent The composition of homomorphisms $(h, \alpha):G\to G'$ and $(h', \alpha'): G'\to G''$ is defined to be $(h'\circ h, \alpha'\circ \alpha):G'\to G''$, 
where $\alpha'\circ \alpha:=\fami {\alpha'_\sigma\circ\alpha_\sigma} \sigma \sortsetS$.

An `injective' homomorphism is called an embedding: 
\begin{definition}[Embeddings]
Let $G$ and $G'$ be formula graphs.
A homomorphism $(h, \alpha):G\rightarrow G'$ is called an \emph{embedding of $G$ into $G'$} if $h$ and $\alpha_\sigma$ for each $\sigma\in \mc S$ are injective. 
A formula graph $G$ \emph{can be embedded into} $G'$ if there is an embedding of $G$ into $G'$. \hfill $\dashv$
\end{definition}
\noindent For example, a \emph{formula subgraph} can be embedded into its original formula graph.
The concept of a formula subgraph is defined as below:
\begin{definition}[Formula Subgraphs]\label{definition: formula subgraphs}
A formula graph $\hypograph{V}{E}$ is a \emph{formula subgraph} (or \emph{subgraph} for short) \emph{of a formula graph} $\hypograph{V'}{E'}$ if $V$ and $E_\syntax{R}$ for each $\syntax{R}\in \mc R$ are subsets of $V'$ and $E'_\syntax{R}$, respectively. \hfill $\dashv$
\end{definition}
\begin{example}
Let $G$ and $G'$ be a formula graph.
Suppose that $G$ is a subgraph of $G'$.
Then, the homomorphism $({\rm inc}, {\rm id}_{Term}): G\to G'$ with  ${\rm inc}$ the inclusion mapping and ${\rm id}_{Term}=({\rm id}_{Term_\sigma})_{\sigma \in \mc S}$ the family of identity mappings is an embedding  of $G$ into $G'$. \checked\checkedC \hfill $\dashv$
\end{example}

An embedding entails logical implication in the converse direction:
\begin{lemma}\label{lemma: inj func to bij substitution}
Let $\sigma\in\sortsetS$ be a sort, $\mc V'_\sigma, \mc V''_\sigma$ be finite subsets of $\mc V_\sigma$ and $\alpha_\sigma:\mc  V'_\sigma\to \mc V''_\sigma$ be an injection.
Define a function $\bar\alpha_\sigma:\FOtermset_\sigma\to \FOtermset_\sigma$ as follows:
\[
\bar \alpha_\sigma(x) =
\begin{cases}
  x  & x\in Term_\sigma - (\mc V'_\sigma\cup \mc V''_\sigma)\\
  \alpha_\sigma(x) & x\in \mc V'_\sigma \\
  \beta_\sigma(x) & x\in (\mc V'_\sigma\cup \mc V''_\sigma) -\mc V'_\sigma,
\end{cases}
\]
where $\beta_\sigma$ is a bijection from $(\mc V'_\sigma\cup \mc V''_\sigma) -\mc V'_\sigma$ to $(\mc V'_\sigma\cup \mc V''_\sigma)-\alpha_\sigma(\mc V'_\sigma)$.
Then, $\bar \alpha_\sigma$ is a bijective substitution of $\sigma$ such that $\bar{\alpha}_\sigma(x)=x$ for any $x\in Term_\sigma - (\mc V'_\sigma\cup \mc V''_\sigma)$ and that $\bar{\alpha}_\sigma(x)=\alpha_\sigma(x)$ for any $x\in \mc V'_\sigma$. \checked\checkedC
\end{lemma}
\begin{proof}
This is easy to prove since $\alpha_\sigma$ is an injection from $\varsetV'_\sigma$ to $\varsetV''_\sigma$.
Function $\beta_\sigma$ exists because  $(\mc V'_\sigma\cup \mc V''_\sigma) -\mc V'_\sigma$ and $(\mc V'_\sigma\cup \mc V''_\sigma)-\alpha_\sigma(\mc V'_\sigma)$ have the same cardinality.\checked\checkedC
\end{proof}

\begin{lemma}\label{InjectiveEmbeddingTransformLemma}\checked\checkedC
Let $G, G'$ be formula graphs.
Suppose that a formula graph $G$ can be embedded into a formula graph $G'$.
Then, there is an embedding $(h, \alpha):G\rightarrow G'$ such that: for any $\sigma\in\mc S$,
\begin{enumerate}
\item
the substitution $\alpha_\sigma$ is bijective; and
\item
there is a finite set $\mc V'_\sigma\subseteq \mc V_\sigma$ of first-order variables satisfying $\alpha_\sigma(x)=x$ for any $x\in Term_\sigma - \mc V'_\sigma$.
\end{enumerate}

\end{lemma}
\begin{proof}
Let $G$ and $G'$ be formula graphs.
Suppose also that $(h, \alpha)$ embeds $G$ into $G'$.
Define $\alpha'_\sigma: Var_\sigma(G)\to \alpha_\sigma(Var_\sigma(G))$ to be the restriction of $\alpha_\sigma$ for each $\sigma \in \sortsetS$.
Note that $Var_\sigma(V)$ and $\alpha_\sigma(Var_\sigma(G))$ are finite sets by Definition \ref{definition: formula graph}.
Hence, by Lemma \ref{lemma: inj func to bij substitution}, $(h, \bar{\alpha'})$ is an embedding of $G$ into $G'$ such that $\bar{\alpha'}_\sigma$ is bijective for any $\sigma\in\sortsetS$ and that $\bar{\alpha'}_\sigma(x)=x$ holds for any $\sigma\in\sortsetS$ and $x\in Term_\sigma -Var_\sigma(G)\cup \alpha_\sigma(Var_\sigma(G))$.
Since $Var_\sigma(G)\cup \alpha_\sigma(Var_\sigma(G))$ is a finite subset of $\varsetV_\sigma$, this completes the proof.

\checked\checkedC
\end{proof}

\begin{proposition}\proved\label{prop: hom induces implication}
Let  $G$  and $G'$ be formula graphs.
If  $G$ can be embedded into $G'$, then $Sent(G')$ logically implies $Sent(G)$ in any class of structures.\checked\checkedC
\end{proposition}
\begin{proof}
Let $G$ and $G'$ be formula graphs.
Suppose that  $G$ can be embedded into $G'$.
Then, there is an embedding $(h, \alpha):G\rightarrow G'$ whose substitutions $\alpha_\sigma$ $(\sigma\in\sortsetS)$ are bijective by Lemma \ref{InjectiveEmbeddingTransformLemma}.
To prove that $Sent(G')$ logically implies $Sent(G)$ in any class of structure,
let $\structure M I$ be a structure and $\mu_1=\fami {\mu_\sigma} \sigma \sortsetS$ and $\mu_2$ be a first-order assignment and a second-order one, respectively.
Let us suppose that $\structure M I \models Sent(G')[\mu_1, \mu_2]$.
Then, we have $\structure M I \models Form(G')[\mu'_1, \mu'_2]$ for some $\mu'_1=\fami {\mu'_\sigma} \sigma \sortsetS$ and $\mu'_2$. 
Thus, we see that $\structure M I \models Form(G)[\fami {\mu'_\sigma\circ {\alpha'_\sigma}^{-1}} \sigma \sortsetS, \mu'_2]$ in a straightforward way. 
Therefore, we obtain $\structure M I \models Sent(G)[\mu_1, \mu_2]$ since $Sent(G)$ is the existential closure of $Form(G)$.\checked\checkedC
\end{proof}

The concept of an isomorphism is defined as below:
\begin{definition}[Isomorphisms]\label{definition: isomorphism}
Let $G$ and $G'$ be formula graphs.
A homomorphism $(h, \alpha):G\rightarrow G'$ is called an \emph{isomorphism} if it satisfies:
\begin{enumerate}
\item
$h$ and $\alpha_\sigma$ for each $\sigma$ are bijective; and 
\item
their inverse mappings constitute a homomorphism $(h^{-1}, (\alpha_\sigma^{-1})_{\sigma\in\mc S} ):G'\to G$.
\end{enumerate}
Formula graphs $G$ and $G'$ are \emph{isomorphic} (notation: $G\simeq G'$) if there is an isomorphism between them. \hfill $\dashv$
\end{definition}
\begin{example}
Let $G =(V, \fami {R_\syntax{R}} {\syntax R} \relsetR)$ be a formula graph and $\mathrm{id}_V$ and $\mathrm{id}_{Term_\sigma}$ ($\sigma\in\sortsetS$) be the identity mappings on $\mathrm{id}_V$ and $Term_\sigma$, respectively.
Then, $\mathrm{id}_G:=(\mathrm{id}_V, \fami {\mathrm{id}_{Term_\sigma}} {\sigma} \sortsetS))$ is an isomorphism.  \checked\checkedC \hfill $\dashv$
\end{example}
\noindent By Proposition \ref{prop: hom induces implication}, we obtain the following immediately:
\begin{proposition}
If two formula graphs $G$ and $G'$ are isomorphic,
then $Sent(G)$ and $ Sent(G')$ are logically equivalent in any class of structures.\checked\checkedC
\end{proposition}
\noindent The following proposition plays an important role when proving the halting property of a user-feedback dialogue protocol on hypothesis graphs:
\begin{proposition}\label{proposition: size n formula graphs finite}\proved
There are at most finitely many formula graph of order $n$ up to isomorphism for any non-negative integer $n$.\checked\checkedC
\end{proposition}
\begin{proof}
We prove that any formula graph of order $n$ is isomorphic to one containing only restricted variables.
Let $Max_p$ be the maximum arity of first-order predicates in $\propsetP$. Notice that the set $\mc P$ of first-order predicates is finite and thus there exists the maximum arity.
Then, any formula graph of order $n$ contains at most $(n\cdot Max_p)$-many variables.
Thus,  for each $\sigma\in\sortsetS$, we take a set $V_\sigma$ of first-order variables of $\sigma$ whose size is $n\cdot Max_p$.
Note that there are at most finitely many formula graphs of order $n$ containing only variables of $\sigma$ in $V_\sigma$ for each $\sigma$ since $\sortsetS$, $\mc P$ and $\mc R$ are finite and since an alphabet does not contain function symbols.
Let $G=\hypograph {V} {E})$ be an arbitrary formula graph of order $n$.
Then, for any $\sigma\in\sortsetS$, we have $\cardinal{Var_\sigma(G)}\leq n\cdot Max_p$. Thus, there is an injection $\alpha_\sigma: Var_\sigma(G) \to V_\sigma$, which extends to a bijective substitution $\bar{\alpha}_\sigma:Term_\sigma \to Term_\sigma$ by Lemma \ref{lemma: inj func to bij substitution}.
Let us write $\bar{\alpha}$ to mean $ \fami {\bar{\alpha}_\sigma} \sigma \sortsetS$.
We then define a formula graph $G'=\hypograph {V'} {E'}$ by 
$V'=\{\bar{\alpha}(l)\mid l \in V\}$ and
$E'_\syntax{R}=\{(\bar{\alpha}(l_1),\ldots, \bar{\alpha}(l_n) ) \mid (l_1,\ldots, l_n)\in E_\syntax{R} \}$ for any $\syntax{R}\in \relsetR$.
Here, $\bar{\alpha}(l)$ is defined to be $\syntax{p}(\bar{\alpha}_{\sigma_1}(\syntax{t}_1), \ldots, \bar{\alpha}_{\sigma_n}(\syntax{t}_n))$ (respectively, $\syntax{\neg p}(\bar{\alpha}_{\sigma_1}(\syntax{t}_1), \ldots, \bar{\alpha}_{\sigma_n}(\syntax{t}_n))$) for any first-order literal $l=\syntax{\neg p}(\syntax{t}_1,\ldots, \syntax{t}_n)\in L$ (respectively, $l=\syntax{\neg p}(\syntax{t}_1,\ldots, \syntax{t}_n)\in L$ ), where $(\sigma_1,\ldots, \sigma_n)$ is the sorts of the arguments of $\syntax p$.
We also define $h:V\to \{\bar{\alpha}(l)\mid l \in V\}$ by $h(l)=\bar\alpha(l)$ for any $l\in V$.
Then, $(h, \bar\alpha)$ is an isomorphism from $G$ to $G'$ with $(h^{-1}, \fami {{\bar\alpha_\sigma}^{-1}} \sigma \sortsetS)$ its inverse.\checked\checkedC
\end{proof}

\begin{remark}
The formula graphs and homomorphisms between them constitute a category in the sense of category theory (cf. \cite{maclane1978categories}) \checked\checkedC and the definition of isomorphism (Definition \ref{definition: isomorphism}) is equivalent to the categorical definition of isomorphism in this category.\checked\checkedC
In this sense, the above definitions of homomorphism and isomorphism are natural.
\hfill $\dashv$
\end{remark}

We conclude this section by showing a result on the relationship between the concept of embedding and that of isomorphism.
\begin{lemma}\label{lemma: same domain codomain of substitution}\checked\checkedC
Let $\sigma\in\sortsetS$ be a sort, $\mc V'_\sigma\subseteq \mc V_\sigma$ be a finite set of first-order variables of 
$\sigma$ and $f$ be a bijective substitution of $\sigma$.
Suppose that  $f(x)=x$ holds for any $x\in Term_\sigma -\mc V'_\sigma$.
Then, we have $f(\varsetV'_\sigma)=\varsetV'_\sigma$.
\end{lemma}
\begin{proof}
To prove $f(\varsetV'_\sigma)\subseteq \varsetV'_\sigma$,
we take an arbitrary element $y\in f(\varsetV'_\sigma)$.
Then, we have $y=f(x)$ for some $x\in\varsetV'_\sigma$.
If we have $x=y$, we immediately obtain $y\in f(\varsetV'_\sigma)$.
Thus, we suppose that $x\neq y$.
Then, we see that $f(y)\neq y$ because $f(y)=y=f(x)$ and thus $x=y$ otherwise, which is a contradiction.
Therefore, we have $y\in f(\varsetV'_\sigma)$.
Hence, we obtain $f(\varsetV'_\sigma)\subseteq \varsetV'_\sigma$.
To prove $f(\varsetV'_\sigma)=\varsetV'_\sigma$, we use $\cardinal{\varsetV'_\sigma} = \cardinal {f(\varsetV'_\sigma)}\lneq \infty$, which is proved by the facts that $\varsetV'_\sigma$ and thus $f(\varsetV'_\sigma)$ are finite and that $f$ is bijective.
Hence, we obtain $f(\varsetV'_\sigma)=\varsetV'_\sigma$ from $f(\varsetV'_\sigma)\subseteq \varsetV'_\sigma$ and $\cardinal{\varsetV'_\sigma} = \cardinal {f(\varsetV'_\sigma)}\lneq \infty$.\checked\checkedC
\end{proof}

\begin{lemma}\label{lem: finitely many bijective homs}\checked\checkedC
Let $G$ be a formula graph and $\mc V'_\sigma\subseteq \mc V_\sigma$ a finite set of first-order variables for each sort $\sigma\in\sortsetS$.
Then, there are at most finitely many embeddings $(h, \alpha): G\rightarrow G$ such that, for each $\sigma\in\sortsetS$,
\begin{enumerate}
\item \label{condition: alpha sigma bijection}
$\alpha_\sigma$ is bijective;
\item \label{condition: identitiy if not in V}
 $\alpha_\sigma(x)=x$ holds for any $x\in Term_\sigma -\mc V'_\sigma$.
\end{enumerate}
\end{lemma}
\begin{proof}
Let $\sigma$ be an arbitrary sort in $\sortsetS$.
Define $A$ be the sets of bijective substitutions $f:\FOtermset_\sigma\to \FOtermset_\sigma$ such that $f(x)=x$ for any $x \in Term_\sigma - \varsetV'_\sigma$ 
and $B$ to be the set of bijective functions $g:\varsetV'_\sigma\to \varsetV'_\sigma$.
We first show that these two sets have the same cardinality.
Define function $F:A\to B$ by $F(f):\varsetV'_\sigma\to\varsetV'_\sigma$ to be the restriction of $f$, which is well-defined due to Lemma \ref{lemma: same domain codomain of substitution}.
Define also function $G:B\to A$ so that it maps a function $g\in B$ to another $G(g)$ such that, for any $x\in Term_\sigma$, $G(g)(x)=g(x)$ if $x\in \varsetV'_\sigma$ and $G(g)(x)=x$ otherwise.
Then, we can easily see that $G\circ F$ and $F\circ G$ are the identity mappings on $A$ and $B$, respectively.
Thus, $F$ and $G$ are bijective and thus $A$ and $B$ have the same cardinality.

Let us next prove that there are at most finitely many embeddings $(h, \alpha): G\rightarrow G$ such that, for each $\sigma\in\sortsetS$,
$\alpha_\sigma$ satisfies Condition \ref{condition: alpha sigma bijection} and \ref{condition: identitiy if not in V}.
It is easy to see that $B$ is finite and thus so is $A$.
Therefore, there are at most finitely many families $\fami {\alpha_\sigma} {\sigma} \sortsetS$ such that, for each $\sigma \in \sortsetS$, $\alpha_\sigma$ satisfies Condition \ref{condition: alpha sigma bijection} and \ref{condition: identitiy if not in V}, since $\sortsetS$ is finite. 
We also see that there are at most finitely many $h:V\to V$, where $V$ is the set of vertices of $G=\hypograph V E$.
Hence, there are at most finitely many embeddings $(h, \alpha):G\to G$ such that, for each $\sigma \in \sortsetS$, function $\alpha_\sigma$ satisfies Condition \ref{condition: alpha sigma bijection} and \ref{condition: identitiy if not in V}. \checked\checkedC
\end{proof}

\begin{proposition}\proved\label{proposition: not-iso implies not-embeddable}\checked\checkedC
Let $G$ and $G'$ be formula graphs.
If $G$ and $G'$ can be embedded into each other, then $G$ and $G'$ are isomorphic.
\end{proposition}
\begin{proof}
Let $(f, \alpha):G\rightarrow G'$ and $(g, \beta):G'\rightarrow G$ be embeddings.
By Lemma \ref{InjectiveEmbeddingTransformLemma},
we may also assume for any $\sigma\in\mc S$ that $\alpha_\sigma$ and $\beta_\sigma$ are bijective and that there are finite sets $\mc V'_\sigma, \mc V''_\sigma \subseteq \mc V_\sigma$ ($\sigma\in\mc S$) such that $\alpha_\sigma(x)=x$ for $x \in \FOtermset_\sigma - \varsetV'_\sigma$ and $\beta_\sigma(x)=x$ for any $x \in \FOtermset_\sigma - \varsetV''_\sigma$.
Thus, we may further assume for each $\sigma\in \sortsetS$ that there is a finite set $\mc V'''_\sigma \subseteq \mc V_\sigma$, say $\varsetV'_\sigma\cup \varsetV''_\sigma$, such that $\alpha_\sigma(x)=x$ and $\beta_\sigma(x)=x$ hold for any $x \in \FOtermset_\sigma - \varsetV'''_\sigma$. 

Let us consider, for $n=1,2,\ldots$,  the composite embeddings  $(g\circ f, \beta\circ\alpha)^n =  (g\circ f, \beta\circ\alpha)  \circ \cdots \circ (g\circ f, \beta\circ\alpha)$ of $n$-many  $(g\circ f, \beta\circ\alpha)$.
Then, by Lemma \ref{lemma: same domain codomain of substitution}, for each $n$ and $\sigma\in \sortsetS$, the function $((\beta\circ \alpha)^n)_\sigma$ is a bijective substitution of $\sigma$ such that $((\beta\circ \alpha)^n)_\sigma(x)=x$ holds for any $x \in \FOtermset_\sigma - \varsetV'''_\sigma$.
Hence, there are at most finitely many such embeddings $(g\circ f, \beta\circ\alpha)^n$ and thus $(g\circ f, \beta\circ\alpha)^n = (g\circ f, \beta\circ\alpha)^{n+m}$ holds for some $n\gneq {0}$ and $m\gneq 0$.
Here, note that the first component $g\circ f$ of $(g\circ f, \beta\circ\alpha)$ is an injective endofunction on a finite set and thus it is bijective.
Note also that the component $\alpha_\sigma\circ\beta_\sigma$ for each $\sigma\in \mc S$ is also bijective.
Therefore, we obtain $ (g\circ f, \beta\circ\alpha)^m={\rm id}_G$ by composing $((g\circ f)^{-1}, \fami {(\beta_\sigma\circ\alpha_\sigma)^{-1}} \sigma \sortsetS )^n$ on the both sides of $(g\circ f, \beta\circ\alpha)^n = (g\circ f, \beta\circ\alpha)^{n+m}$.
By the symmetric argument, we also have $(f\circ g, \alpha\circ\beta)^{m'}={\rm id}_{G'}$ for some $m'\gneq 0$.
Hence, we have $(g\circ f, \beta\circ\alpha)^{mm'}={\rm id}_G$ and $(f\circ g, \alpha\circ\beta)^{mm'}={\rm id}_{G'}$.
Therefore, $(f, \alpha)$ and $(g, \beta)$ are an isomorphism.\checked\checkedC
\end{proof}

\section{User-Feedback Dialogue Protocols}\label{section: feedback protocols}
In the previous section, we have defined the concepts of hypothesis and hypothesis graph.
In this section, we first propose the concept of a \emph{user-feedback dialogue} (\emph{UFBD}) \emph{protocol} in Section \ref{subsection: Abstract User-Feedback Dialogue Protocols}.
We then define two types of UFBD protocols, \emph{basic UFBD protocols} and \emph{simple UFBD protocols}, as special cases of UFBD protocols in Sections \ref{subsection: Basic UFBD protocols} and \ref{section: simple UFBD}, respectively.
UFBD protocols on hypotheses and hypothesis graphs are obtained by instantiating UFBD protocols.
For these protocols, we see whether or not the two important properties hold:
\begin{inparaenum}[(1)] \item any dialogue always terminates, \emph{the halting property}; \item if the user gives feedbacks on properties based on the properties of fixed targets, the dialogue achieves the targets if it terminates, \emph{the convergence property}.
\end{inparaenum}

\subsection{Generic Results on User-Feedback Dialogue Protocols}\label{subsection: Abstract User-Feedback Dialogue Protocols}
We shall define the concept of a user-feedback dialogue (UFBD) protocol as a set of dialogues between the user and the system.
It is defined on the basis of the concept of \emph{feedback on properties}.
Let $\mc X$ be a set of \emph{candidates}, which shall be instantiated as a set of hypotheses/hypothesis graphs, and $P$ a set of \emph{properties}. 
We call a mapping $Prop:\mc X\to \wp(P)$ a \emph{property assignment function}.
Pairs  $(p, pos)$, $(p, neg)$ and $(p, neutral)$ for any $p\in P$ are called \emph{feedbacks (FBs) on $P$}.
Using these concepts, the concept of a UFBD protocol is defined:
\begin{definition}[User-Feedback Dialogues]\label{definition: UFBD}
Let $\mc X$ be a set of candidates, $P$ be a set of properties and $Prop:\mc X\to \wp(P)$ be a property assignment function.
A \emph{user-feedback dialogue} (\emph{UFBD}) on $\property$ is a finite or an infinite sequence $\mc D=\mc X_1, \mc F_1, \mc X_2, \mc F_2,\dots, \mc X_i, \mc F_i, \ldots$ of a set $\mc X_i\subseteq \mc X$ and a set $\mc F_i$ of FBs on $P$ ($i=1,2,\ldots$) that satisfies the following two conditions for any index $i$:
\begin{enumerate}
\item \label{condition: satisfy all previous FBs}
 \camera{any \emph{$X\in \mc X_i$ satisfies all previous FBs}, i.e for any} $j\lneq i$,
\begin{enumerate}
\item
$p \in Prop(X)$ for any  $ (p, pos) \in \mc F_j$  and
\item
 $p \not\in Prop(X)$ for any $ (p, neg) \in\mc F_j$;
\end{enumerate}
\item \label{condition: non-empty X}
$\mc X_i$ is non-empty if there is a candidate in $\mc X$ that \camera{satisfies all previous FBs.}
\end{enumerate}
A \emph{user-feedback dialogue protocol} (\emph{UFBD protocol} for short) $\class D$ is a set of UFBDs on $\property$ such that,
 for any finite or infinite sequence $\mc D=\mc X_1, \mc F_1,\dots \in \class D$, set $\class D$ contains any successive sequence $\mc X_1, \mc F_1,\dots, \mc X_i$/$\mc X_1, \mc F_1,\dots, \mc X_i, \mc F_i$ of $\mc D$ starting from $\mc X_1$, including the null sequence.
We call the set of UFBDs on $\property$ \emph{the UFBD protocol on $\property$} and write $\protocolProp \UFBD \property$.
\end{definition}
\noindent We often call a UFBD a \emph{dialogue} if $\property$ is clear from the context. A dialogue is \emph{finite} if it is a finite sequence.
Note that both of $\mc X_i$ and $\mc F_i$ possibly appear as the last set of a UFBD.

Intuitively, $\mc X$ in a UFBD is the set of possible hypotheses/hypothesis graphs.
The abductive reasoner is required to present a set of hypotheses/hypothesis graphs in $\mc X$ that satisfies all previous FBs. The user gives feedback on properties. 
To illustrate this, we give an example:
\begin{example}\label{example: UFBD}\proved\checked
Let $\mc X$ be the set $\{G_0, G_1, G_2\}$, where $G_0$ is the hypothesis graph in Example \ref{example: MFM graph} and $G_1, G_2$ are those defined below:
\begin{figure}[h]
\begin{minipage}{0.5\columnwidth}
\centering
{\scriptsize
\begin{tikzpicture}[node distance =0cm and 0.7cm]
\node [name=q00] {$\syntax{Hold(F, No)}$};
\node [below = of q00] (q01) {$\syntax{Open(F)}$};
\coordinate [right = of q00, label=above:$\syntax{Action}$] (q10) ;
\node [right = of q10] (q20) {$\syntax{Hold(F, High)}$};
\draw [-] (q01) -| (q10);
\draw [->] (q00) -- (q20);\
\end{tikzpicture}\\
Hypothesis graph $G_1$
}
\end{minipage}
\begin{minipage}{0.5\columnwidth}
\centering
{\scriptsize
\begin{tikzpicture}[node distance =0cm and 0.7cm]
\node [name = q20] {$\syntax{Hold(F, High)}$};
\node [below = of q20, white] (q21) {O};
\coordinate [right = of q20, label=above:\texttt{Cause-Effect}] (q30) ;
\node [right = of q30] (q40) {$\syntax{Hold(P, High)}$};
\draw [->] (q20) -- (q40);
\end{tikzpicture}\\
Hypothesis graph $G_2$}
\end{minipage}
\end{figure}\\
\noindent Let also $P$ be $\formgraphset$ and $Prop:\mc X\to \wp(P)$ be the function defined by $Prop(G)=\{G' \mid G' \text{ is a subgraph of }G\}$ for any $G\in \mc X$.
Then, the sequence  $\mc X_1:=\{G_1, G_2\}, \mc F_1:=\{(G_1, pos), (G_2, pos)\}$, $ \mc X_2:=\{G_0\}$ is a UFBD on $Prop$.
This sequence is intended to appear as follows:
\begin{enumerate}
\item
$\mc X$ is the set of possible hypothesis graphs;
\item
the abductive reasoner generates and presents, to the user, hypothesis graphs $G_1$ and $G_2$ in $\mc X$ as $\mc X_1$;
\item
the user gives FBs $(G_1, pos)$ and $(G_2, pos)$ as $\mc F_1$, which require that $G_1$ and $G_2$  should be contained as subgraphs;
\item
the reasoner regenerates and presents as $\mc X_2$ the only graph $G_0$ in $\mc X$ that contains $G_1$ and $G_2$ as its subgraphs.\hfill$\dashv$
\end{enumerate}
\end{example}

In the rest of this article, we see whether or not specific user-feedback protocols have \emph{the halting property} and \emph{the convergence property}.
\begin{definition}[Halting Property]
A user-feedback protocol has \emph{the halting property} if it does not contains any infinite sequence. \hfill$\dashv$
\end{definition}
\noindent The convergence property is defined using the concept of target set $\targetsetX\subseteq \mc X$.

\begin{definition}\label{definition: towards}
Let $\property:\mc X \to \powerset P$ be a property assignment function, $\mc D=\mc X_1, \mc F_1, \mc X_2, \mc F_2,\dots, \mc X_i, \mc F_i, \ldots$ be a finite or an infinite in $\protocolProp \UFBD \property$ and $\targetsetX$ be a non-empty subset of $\mc X$.
Dialogue $\mc D$ is \emph{towards $\targetsetX$} if it satisfies the following conditions for each index $i$ and $(p, f)\in \mc F_i$:
\begin{enumerate}
\item
$f=pos$ if $p\in Prop(X)$ for any $X\in \targetsetX$,

\item
$f=neg$ if $p \notin Prop(X)$ for any $X\in \targetsetX$  and 

\item $f=neutral$ otherwise.
\end{enumerate}
If the target set $\targetsetX$ is a singleton $\{X_t\}$, we call a UFBD on $\property$ towards $X_t$ to mean a UFBD on $\property$ towards $\{X_t\}$. \hfill$\dashv$
\end{definition}
\noindent For a UFBD protocol $\class D$ on $\property$ and a non-empty subset $\targetsetX$ of $\mc X$, 
 we write $\protocolTowards {\class D} \targetsetX$ to mean the subset of $\class D$ whose dialogues are towards  $\targetsetX$.
Note that $\protocolTowards {\class D} \targetsetX$ is again a UFBD protocol.\checked\checkedC

Let us next define the concept of convergence and the convergence property.
\begin{definition}[Convergence]\label{definition: convergence of one dialogue}
Let $\property:\mc X\to \powerset P$ be a property assignment function, $\class D$ be a UFBD protocol and $\targetsetX$ be a non-empty subset of $\mc X$.
A finite dialogue $\mc D$ \emph{converges to $\targetsetX$ in $\class D$} if it satisfies the condition: if there is no set $\mc A$ satisfying that the sequence $\mc D, \mc A$ is a dialogue in $\class D$,
then:
\begin{enumerate}
\item \label{condition: last set in convergence}
the last set of $\mc D$ is a non-empty subset of $\mc X$;
\item \label{condition: satisfy shared properties}
for any $X$ in the last set of $\mc D$ and $p\in P$, we have
\begin{enumerate}
\item \label{condition: satisfy shared positive properties}
$p\in Prop(X)$ if $p\in Prop(X')$ for any $X'\in \targetsetX$

\item \label{condition: satisfy shared negative properties}
$p\notin Prop(X)$ if $p\notin Prop(X')$ for any $X'\in \targetsetX$.
\end{enumerate}
\end{enumerate}
\hfill$\dashv$
\end{definition}

\begin{definition}[Convergence Property]\label{definition: convergence}
Let $\property:\mc X\to \powerset P$ be a property assignment function and $\class D$ be a UFBD protocol on $\property:\mc X\to \powerset P$.
Protocol $\class D$ has \emph{the convergence property} (respectively, \emph{the convergence property towards a single target}) if it satisfies the following conditions:
for any non-empty set (respectively, singleton) $\targetsetX\subseteq \mc X$ and finite dialogue $\mc D$ in $\protocolTowards {\class D} \targetsetX$,  $\mc D$ converges to $\targetsetX$ in $\protocolTowards {\class D} \targetsetX$.
\end{definition}

If the target set $\targetsetX$ is a singleton, we have the following results.
Firstly, any dialogue towards $\targetsetX$ does not contain $neutral$ feedbacks.
\begin{proposition}\label{proposition: single target implies no neutral}\proved
Let $\property:\mc X\to \powerset P$ be a property assignment function, $\targetsetX\subseteq \mc X$ be a singleton and
 $\mc D=\mc X_1, \mc F_1, \ldots $ be a finite or an infinite UFBD on $\property$ towards $\targetsetX$.
Then, for any index $i$ and FB $(p,f)\in \mc F_i$, we have $f\neq neutral$.\checked\checkedC
\end{proposition}
\begin{proof}
Suppose that $\targetsetX = \{\targetX \}$.
If we have $p\in \property(\targetX)$, then $p\in \property(X)$ for any $X\in \targetsetX$, which implies $f=pos$ and thus $f\neq neutral$.
Otherwise, we have $p\notin \property(\targetX)$ and thus $p\notin \property(X)$ for any $X\in \targetsetX$.
This implies $f=neg$. Therefore, we obtain $f\neq neutral$.\checked\checkedC
\end{proof}
\noindent Secondly, when a dialogue towards a singleton $\targetsetX=\{\targetX\}$ in a UFBD protocol satisfying the convergence property terminates, the last set is a set of candidates that have the same properties as those of the target $\targetX$:
\begin{proposition}\proved\label{proposition: convergence single target}\checked\checkedC
Let $\mc X$ be a set of candidates, $P$ be a set of properties, $\property:\mc X\to \wp(P)$ be a property assignment function and
 $\class D$ be a UFBD protocol on $\property$.
 Suppose that $\class D$ has the convergence property.
 For a finite UFBD $\mc D\in\class D$ towards a candidate $X_t\in \mc X$,
 if there is no set $\mc A$ such that $\mc D, \mc A$ is a UFBD on $\property$ towards $X_t$, then:
\begin{enumerate}
\item \label{condition: last set in convergence for a single target}
the last set of $\mc D$ is a non-empty subset of $\mc X$;
\item \label{condition: satisfy shared properties}
we have $\property(X)=\property(X_t)$ for any $X$ in the last set $\lastsetXL$ of $\mc D$.
\end{enumerate}
\end{proposition}
\begin{proof}
Condition \ref{condition: last set in convergence for a single target} is trivial by Condition \ref{condition: last set in convergence}  of Definition \ref{definition: convergence of one dialogue}.
Let us prove Condition \ref{condition: satisfy shared properties}.
To prove the contraposition of $\property(X)\subseteq \property(X_t)$, suppose that $p\notin \property(X_t)$.
Then, we see that $p\notin \property(X)$ for any $X\in \lastsetXL$ by Condition \ref{condition: satisfy shared negative properties} of Definition \ref{definition: convergence of one dialogue}.
To prove the converse direction, suppose that $p\in \property(X_t)$. This implies $p\in\property(X)$  for any $X\in \lastsetXL$ by Condition \ref{condition: satisfy shared positive properties} of Definition \ref{definition: convergence of one dialogue}.
Hence, we obtain $\property(X)=\property(X_t)$ for any $X\in\lastsetXL$.\checked\checkedC
\end{proof}

\subsection{Basic UFBD Protocols}\label{subsection: Basic UFBD protocols}
In the previous section, we have defined the concept of user-feedback dialogue protocol.
In this section, we impose additional condition on $\mc X_i$ and $\mc F_i$ in a UFBD $\mc D=\mc X_1,\mc F_1,\ldots$ to obtain another type of protocols, called \emph{basic UFBD protocols}, and see that they satisfy the halting and the convergence properties under certain reasonable conditions.
\subsubsection{Generic Results on Basic UFBD Protocols}
We first define the concept of basic UFBD protocol and see generic results on the protocols.
The definition is given as below:
\begin{definition}\label{definition: basic UFBD}
Let $\property: \mc X \to \powerset P$ be a property assignment function.
A \emph{basic user-feedback dialogue} (\emph{basic UFBD}) on $\property$ is a finite or an infinite UFBD $\mc D=\mc X_1, \mc F_1, \mc X_2, \mc F_2,\dots, \mc X_i, \mc F_i, \ldots$ on $\property$ that satisfies the conditions below for each index $i$.
\begin{enumerate}
\item Conditions on $\mc X_i$.
\begin{enumerate}
\item \label{condition: un-poited-out property}
There are $X \in \mc X_i$ and $p\in \property(X)$ such that $p$ \emph{has not been pointed out before $i$}, i.e $p\neq p'$ for any $j \lneq i$ and $(p', f')\in\mc F_j$, if there is $X\in \mc X$ such that $X$ satisfies all previous FBs and that some $p\in \property(X)$ has not been pointed out before $i$.
\item \label{condition: last X}
If there is no candidate $X \in \mc X_i$ such that $p$ has not been pointed out before $i$ for some $p\in Prop(X)$, it is the last set of $\mc D$

\end{enumerate}
\item Conditions on $\mc F_i$.
\begin{enumerate}
\item \label{condition2a: formula-based}
For any $(p, f)\in \mc F_i$, \camera{\emph{$p$ appears in $\mc X_i$}, i.e} there is a candidate $X\in \mc X_i$ such that $p \in Prop(X)$.
\item \label{condition: no previous FB}
For any $(p, f)\in \mc F_i$, $p$ has not been pointed out before $i$. 

\item \label{condition: non-empty FBs}
$\mc F_i$ is non-empty if there is a property $p\in P$ that appears in $\mc X_i$ and that has not been pointed out before $i$.
\item \label{condition: empty FB is the last}
If $\mc F_i$ is empty, it is the last set of $\mc D$.
\end{enumerate}

\end{enumerate}
A \emph{basic UFBD protocol on $\property$} is defined to be a set of basic UFBDs on $\property$.
\emph{The basic UFBD protocol on $\property$} (notation: $\protocolProp {\protocolname{Basic}} \property$) is the set of basic UFBDs  on $\property$.
\end{definition}
\noindent Intuitively, in a basic UFBD protocol, the abductive reasoner is required to present, if possible, at least one hypothesis/hypothesis graph in $\mc X$ such that it satisfies all previous FBs and that some property of it has not been pointed out; the user is required to give, if possible, at least one FB on a property of the presented hypothesis/hypothesis graph that has not been pointed.
Example \ref{example: UFBD} is also an example of a basic UFBD and illustrates this.

We see below that $\protocolProp {\protocolname{Basic}} \property$ has the convergence property for any $\property$.
\begin{lemma}\proved \label{lemma: next FB}
Let $P:\mc X\to \powerset P$ be a property assignment function, $\targetsetX$ be a non-empty subset of $\mc X$ and $\mc D=\mc X_1,\mc F_1, \ldots, \mc X_i$ $(i\geq 1)$ be a finite dialogue in $\protocolTowards {\protocolProp {\protocolname{Basic}} \property} \targetsetX$.
Suppose that there is a candidate $X\in \mc X_i$ and a property $p\in \property(X)$ such that $p$ has not been pointed out before $i$.
Then, there is a non-empty set $\mc F_i$ of FBs on $P$ such that the sequence $\mc D, \mc F_i$ is again a dialogue in $\protocolTowards {\protocolProp {\protocolname{Basic}} \property} \targetsetX$.\checked\checkedC
\end{lemma}
\begin{proof}
We first construct $\mc F_i=\{(p, f)\}$ as follows:
 $f=pos$ if $p\in \property(X)$ for any $X\in\targetsetX$;   $f=neg$ if $p\notin \property(X)$ for any $X\in\targetsetX$;    $f=neutral$ otherwise.
Then, it is easy to see that $\mc D, \mc F_i$ is a dialogue in $\protocolTowards {\protocolProp {\protocolname{Basic}} \property} \targetsetX$.\checked\checkedC
\end{proof}

\begin{lemma}\label{lemma: exist next X}\proved\checked\checkedC
Let $\property:\mc X\to \powerset P$ be a property assignment function, $\targetsetX$ be a non-empty subset of $\mc X$ and $\mc D=\mc X_1, \mc F_1, \ldots, \mc X_i, \mc F_i$ be a finite dialogue in $\protocolTowards {\protocolProp {\protocolname{Basic}} \property} \targetsetX$.
Then, there exists a set $\mc X_{i+1}\subseteq \mc X$ such that $\mc D, \mc X_{i+1}$ is a dialogue in $\protocolTowards {\protocolProp {\protocolname{Basic}} \property} \targetsetX$.
\end{lemma}
\begin{proof}
Let us first consider the case where $\mc D$ is the null sequence.
If there is $X\in \mc X$ such that $\property(X)\neq \emptyset$, we define $\mc X_1$ to be $\{X\}$.
Otherwise, we define $\mc X_1$ to be any non-empty subset of $\mc X$, say $\targetsetX$.
Then, the sequence $\mc X_1$ is a dialogue in $\protocolTowards {\protocolProp {\protocolname{Basic}} \property} \targetsetX$.
We next consider the case where $\mc D$ is not the null sequence.
In this case, there is a candidate $X\in \mc X_i$ and $p\in P$ such that $p$ has not been pointed out before $i$ by Condition \ref{condition: last X} of Definition \ref{definition: basic UFBD} since $\mc X_i$ is not the last set of $\mc D$.
Hence, $\mc F_i$ is non-empty by Condition \ref{condition: non-empty FBs} of Definition \ref{definition: basic UFBD}.
Therefore, Condition \ref{condition: empty FB is the last} of Definition \ref{definition: basic UFBD} does not apply to $\mc F_i$.
We next give $\mc X_{i+1}\subseteq \mc X$.
If there is a candidate $X\in\mc X$ such that it satisfies all previous FBs and that some property $p\in Prop(X)$ has not been pointed out before $i+1$,
then $\mc D, \{X\}$ is a dialogue in $\protocolTowards {\protocolProp {\protocolname{Basic}} \property} \targetsetX$.
Otherwise, $\mc D, \{X\}$ for any $X\in\targetsetX\neq \emptyset$ is a dialogue in $\protocolTowards {\protocolProp {\protocolname{Basic}} \property} \targetsetX$ because any $X\in\targetsetX$ satisfies all previous FBs.\checked\checkedC
\end{proof}

\begin{theorem}\proved\checked\checkedC\label{theorem: convergence}
Let $\property: \mc X\to \powerset P$ be a property assignment function.
Then, $\protocolProp {\basicUFBD} \property$ has the convergence property.
\end{theorem}
\begin{proof}
We prove each of the conditions in Definition \ref{definition: convergence of one dialogue} for any finite dialogue $\mc D\in \protocolProp {\basicUFBD} \property$ and a non-empty subset $\targetsetX\subseteq \mc X$.
We first prove Condition \ref{condition: last set in convergence}. Suppose, for the sake of contradiction, that the last set of $\mc D$ is not a set of candidates in $\mc X$.
Then, there exists a subset $\mc X'\subseteq \mc X$ such that $\mc D, \mc X'$ is a dialogue in $\protocolTowards {\protocolProp {\basicUFBD} \property} \targetsetX$ due to Lemma \ref{lemma: exist next X}, which is a contradiction.
We can also see that the last set is non-empty by Condition \ref{condition: non-empty X} of Definition \ref{definition: UFBD} since any candidate $X\in\targetsetX\neq\emptyset$ satisfies all previous FBs.

We next prove Condition \ref{condition: satisfy shared properties}.
To prove Condition \ref{condition: satisfy shared positive properties}, let us take an arbitrary $X$ in the last set $\lastsetXL$ of $\mc D$ and $p\in P$ and suppose that $p\in Prop(X')$ for any $X' \in \targetsetX$.
Suppose also, for the sake of contradiction, that $p\notin \property(X)$.
Then, we see that $p$ has not been pointed out before $l$ by Condition \ref{condition: satisfy all previous FBs} in Definition \ref{definition: UFBD} and Definition \ref{definition: towards}.
Therefore, there is $\mc F$ such that $\mc D, \mc F$ is a dialogue in $\protocolTowards {\protocolProp {\basicUFBD} \property} \targetsetX$ due to Lemma \ref{lemma: next FB}. This is a contradiction.
Condition \ref{condition: satisfy shared negative properties} is proved in a similar way.\checked\checkedC
\end{proof}
\noindent By Proposition \ref{proposition: convergence single target} and Theorem \ref{theorem: convergence}, we obtain:
\begin{corollary}\checked\checkedC\label{corollary: the same property at the end basic prop}
Let $\mc D$ be a finite basic UFBD on $\property:\mc X\to \powerset P$ towards $X_t$.
Suppose that there is no set $\mc A$ such that $\mc D, \mc A$ is a UFBD on $\property$ towards $X_t$.
Then, the last set is a non-empty subset of $\mc X$ and we have $\property(X')=\property(X)$ for any $X'$ in the last set of $\mc D$.\checked\checkedC
\end{corollary}

As we shall see in Examples \ref{example: semantics-based protocol not terminate} and \ref{example: infinite syntax-based dialogue}, a basic UFBD doe not necessarily have the halting property.

\subsubsection{Basic UFBD Protocols on Hypotheses}\label{subsec: feedback on Hypotheses}
In the rest of this article, we assume that $\class C$ ranges over classes of structures and $O$ over sets of sentences, except in examples.
A basic UFBD protocol on hypotheses is obtained by substituting the following property assignment function $\propertyH O {\class C}$ for $Prop$:
\begin{definition}
We define a function $\propertyH O {\class C}: \hypset {O} {\class C} \to \wp(\quotient{\sentenceSet}{\leftrightarrow_{\class C}})$ by:
 \[\propertyH O {\class C} (H)=\{[\syntax{\Phi}]_{\leftrightarrow_{\class C}} \mid H \models_{\class C} \syntax{\Phi}\}\]
  for any hypothesis $H\in\hypset {O} {\class C}$.
Here, $\quotient{\sentenceSet}{\leftrightarrow_{\class C}}$ is the set of equivalence classes of $\sentenceSet$ for the logical equivalence relation $\leftrightarrow_{\class C}$ restricted to $\sentenceSet\times\sentenceSet$.
\end{definition}
\noindent This property assignment function characterises a hypothesis up to equivalence:
\begin{proposition}\label{prop: proph and equiv}\proved\checked\checkedC
Let $H, H'$ be hypotheses in $\hypset {O} {\class C} $.
Then, the following are equivalent:
\begin{enumerate}
\item \label{condition: proph and equiv, hyp equiv}
$H\leftrightarrow_{\class C}H'$;
\item \label{condition: proph and equiv, proph equal}
$\propertyH O {\class C}(H)=\propertyH O {\class C}(H')$.
\end{enumerate}
\end{proposition}
\begin{proof}
It is easy to see that Condition \ref{condition: proph and equiv, hyp equiv} implies Condition \ref{condition: proph and equiv, proph equal}. Thus, we here prove the direction from Condition \ref{condition: proph and equiv, proph equal} to Condition \ref{condition: proph and equiv, hyp equiv}.
By Condition \ref{condition: proph and equiv, proph equal}, we have $[\bigwedge H]_{\leftrightarrow_{\class C}}\in \propertyH O {\class C}(H)=\propertyH O {\class C}(H')$.
Hence, we have $H'\models_{\class C}\bigwedge H$ and thus $H'\models_{\class C}H$.
Similarly, we have  $H'\models_{\class C}H$.
Therefore, $H$ and $H'$ are equivalent in $\class C$.\proved\checked\checkedC
\end{proof}
\noindent As an immediate consequence, we obtain:
\begin{corollary}
Let $H_t$ be a hypothesis in $\hypset {O} {\class C}$ and $\mc D$ be a finite dialogue in $\protocolTowards {\protocolProp {\protocolname{Basic}} {\propertyH O {\class C}} } {\{H_t\}}$.
Suppose that there is no set $\mc A$ such that $\mc D, \mc A$ is a dialogue in  $\protocolTowards {\protocolProp {\protocolname{Basic}} {\propertyH O {\class C}} } {\{H_t\}}$.
Then, the last set is a non-empty subset of $\hypset {O} {\class C}$ and  any hypothesis in the last set of $\mc D$ is equivalent to $H_t$ in $\class C$.\checked\checkedC
\end{corollary}
\begin{proof}
This is readily proved by Corollary \ref{corollary: the same property at the end basic prop} and Proposition \ref{prop: proph and equiv}.\checked\checkedC
\end{proof}

The protocol $\protocolProp {\basicUFBD} {\propertyH O {\class C}}$ does not have the halting property for some $\class C$ and $O$:
\begin{example}\label{example: semantics-based protocol not terminate}\proved\checked
Consider a single-sorted language such that $\mc C =\{\syntax{0}\}$, $\mc P=\{\syntax{p}\}$ and $\mc R = \{\syntax{R}\}$.
Let $\class{C}$ be the class of all structures $\structure{M}{I}$ such that the domain $M$ of $\mc M$ is the set $ \mathbb{Z}_{\ge 0}$ of non-negative integers, $\syntax{0}^{\mc M}=0$ and $\syntax{p}^\mc{M}= \mathbb{Z}_{\ge 0}$, and let $O$ be $\{\syntax{p(0)}\}$.
Define a hypothesis $H_0$ to be $\syntax{\{p(0)\land R(p(0),p(0))\}}$ and 
$H_n$ ($n=1,2,\ldots$) to be the existential closure of the conjunction of:
 \begin{gather*}
\syntax{p(0)}, \quad
\bigwedge\{\syntax{p(x}_i) \mid 0\lneq  i \leq n\},\\
\bigwedge\{\syntax{0 \neq x}_i \mid 0\lneq  i \leq n \},\quad
\bigwedge\{ \syntax{x}_i\neq \syntax{x}_j \mid 0\lneq  i \lneq  j \leq n\},\\
\bigwedge\{\syntax{R(p(0), p(}\syntax{x}_1)), \syntax{R(p(x}_{n}), \syntax{p(0))}\},\quad
\bigwedge\{\syntax{R}(\syntax{p(x}_i), \syntax{p(x}_{i +1}))\mid 0\lneq  i\lneq n\}.
\end{gather*}
Note that, for $n,n' =0,1,\ldots$,  $H_n$ and $H_{n'}$ do not logically imply each other in $\class C$ if $n\neq n'$.
Hence, $\mc H_i=\{H_i\}$ and $\mc F_i =\{([H_i]_{\leftrightarrow_{\class C}},neg)\}$ for $i=1,2,\ldots$ give an infinite dialogue $\mc H_1, \mc F_1, \mc H_2, \mc F_2, \ldots$ in $\protocolTowards {\protocolProp {\protocolname{Basic}} {\propertyH O {\class C}} } {\{H_0\}}$  and thus in $\protocolProp {\basicUFBD} {\propertyH O {\class C}}$. \hfill $\dashv$
\end{example}
\noindent Restricting the class $\class C$ of structures, the protocol $\protocolProp {\basicUFBD} {\propertyH O {\class C}}$ has the halting property for any $O\subseteq \sentenceSet$:
\begin{theorem}\proved\checked\checkedC\label{theorem: halting property of syntax-based dialogue protocol}
Suppose that the domain $M_\sigma$, shared by the structures $\structure{M}{I}$ in $\class C$, is finite for any sort $\sigma\in\mc S$.
Then, protocols $\protocolname{Basic}(\propertyH O {\class C})$ and thus $\protocolTowards {\protocolname{Basic}(\propertyH O {\class C})} {\targetsetH}$ for any $\targetsetH \subseteq \hypset O {\class C}$ have the halting property.
\end{theorem}
\begin{proof}
It is easy to see that $\class C$ is finite since $M_\sigma$ is finite for any sort $\sigma\in\mc S$, which implies that $\quotient{\sentenceSet}{\leftrightarrow_{\class C}}$ is also finite.
Assume, for the sake of contradiction, that there is an infinite basic UFBD $\mc D= \mc X_1, \mc F_1, \ldots, $ on $\propertyH O {\class C}$.
Then, we readily see that each $\mc F_i$ is non-empty since if it is empty then it is the last set due to Condition \ref{condition: empty FB is the last} in Definition \ref{definition: basic UFBD}.
Let $S_i$ be the set $\bigcup_{j=1}^{i}\{p \in \quotient{\sentenceSet}{\leftrightarrow_{\class C}} \mid (p, f)\in \mc F_j \}$ for any index $i$.
Then, we have $\cardinal{S_k}\lneq\cardinal{S_{k'}}$ for any indices $k, k'$ with $k \lneq k'$ because there is an FB $(p, f)\in \mc F_{k'}$ that is not contained in $\mc F_{k''}$ for any $k''$ with $k''\lneq k'$ by Condition \ref{condition: no previous FB} in Definition \ref{definition: basic UFBD}.
Hence, we can see that $i \leq \cardinal{S_i}$ by induction on $i$.
Therefore, we have $ \cardinal{ \quotient{\sentenceSet}{\leftrightarrow_{\class C}} } \lneq S_j$ for any index with $\cardinal{ \quotient{\sentenceSet}{\leftrightarrow_{\class C}} } \lneq j$.
However, we also have $\cardinal{S_j}\leq \cardinal{ \quotient{\sentenceSet}{\leftrightarrow_{\class C}} }$ for any index $j$ due to $S_j\subseteq  \quotient{\sentenceSet}{\leftrightarrow_{\class C}}$.
This is a contradiction.\checked\checkedC
\end{proof}

\subsubsection{Basic UFBD Protocols on Hypothesis Graphs}\label{subsec: feedback on Hypothesis Graphs}
In the previous section, we have discussed basic UFBD protocols on hypotheses.
In this section, we discuss basic UFBD protocols on hypothesis graphs.
We first see that a basic UFBD protocol on hypothesis graphs does not necessarily have the halting property and then show that it necessarily does under certain reasonable conditions.

A basic UFBD on hypothesis graphs is formally obtained by instantiating $Prop$ for the following property assignment function:
\begin{definition}
We define a function $\propertyG O {\class C}: \hypgraphset O {\class C}\to \wp(\quotient {\formgraphset} {\simeq})$ by:
\[\propertyG O {\class C} (G)=\{[G']_{\simeq} \mid G' \text{ can be embedded into }G\}\] for any $G\in \hypgraphset O {\class C}$.
Here, $\quotient \formgraphset {\simeq}$ is the equivalence classes of $\formgraphset$ for the relation $\simeq$ of being isomorphic and $[G']_{\simeq}$ is the equivalence class containing $G'$.
\end{definition}
\noindent We define the \emph{order} $\cardinal{g}$ of $g\in \quotient {\formgraphset} {\simeq}$ to be the order $\cardinal{G}$ of some $G\in g$, which is well-defined since $g$ is non-empty and since any two hypothesis graphs in $g$ are isomorphic.
This property assignment function $\propertyG O {\class C}$ characterises a hypothesis graph up to isomorphism:
\begin{proposition}\label{prop: prop and isom}\proved\checked\checkedC
For any $G, G'\in\hypgraphset O {\class C}$, the following are equivalent:
\begin{enumerate}
\item \label{condition: isom G G'}
$G\simeq G'$;
\item \label{condition: same prop G G'}
$PropG(O, {\class C})(G)=\propertyG O {\class C} (G')$.
\end{enumerate}
\end{proposition}
\begin{proof}
We first prove the direction from Condition \ref{condition: isom G G'} to Condition \ref{condition: same prop G G'}.
Let $(i, I):G\to G'$ be an isomorphism.
Then, for any formula graph $G''$ and embedding $(e, E):G'' \to G$, the composition $(i, I)\circ (e, E):G'' \to G'$ is also an embedding.
Therefore, we have that $\propertyG O {\class C} (G)\subseteq \propertyG O {\class C} (G')$.
Similarly, we also have $\propertyG O {\class C} (G)\supseteq \propertyG O {\class C} (G')$.
Hence, we obtain $\propertyG O {\class C} (G)= \propertyG O {\class C} (G')$.
Let us next prove the converse direction. 
Using the fact that $[G]_{\simeq} \in PropG(O, {\class C})(G)=\propertyG O {\class C} (G') $, we can easily see that $G$ can be embedded into $G'$.
Similarly, we can also see that $G'$ can be embedded into $G$.
Hence, $G$ and $G'$ can be embedded into each other.
Therefore, $G$ and $G'$ are isomorphic by Proposition \ref{proposition: not-iso implies not-embeddable}.
\checked\checkedC
\end{proof}
{\noindent}As an immediate consequence, we obtain:
\begin{corollary}\checked\checkedC
Let $G_t$ be a hypothesis graph in $\hypgraphset O {\class C}$ and $\mc D$ be a finite dialogue in $\protocolTowards {\protocolProp {\protocolname{Basic}} {\propertyG O {\class C}} } {\{G_t\}}$.
Suppose that there is no set $\mc A$ such that $\mc D, \mc A$ is a dialogue in $\protocolTowards {\protocolProp {\protocolname{Basic}} {\propertyG O {\class C}} } {\{G_t\}}$.
Then, the last set is a non-empty subset of $\hypgraphset {O} {\class C}$ and any hypothesis graph in the last set of $\mc D$ is isomorphic to $G_t$ .
\end{corollary}
\begin{proof}
This is proved by Corollary \ref{corollary: the same property at the end basic prop} and Proposition \ref{prop: prop and isom}.\checked\checkedC
\end{proof}

The protocol $\protocolProp {\protocolname{Basic}} {\propertyG O {\class C}} $ does not have the halting property for some $O$ and $\class C$:
\begin{example}\label{example: infinite syntax-based dialogue}
We use the settings in Example \ref{example: semantics-based protocol not terminate} and define $G_0$ to be the hypothesis graph  $(V^{(0)}, E^{(0)}_\syntax{R})=(\{\syntax{p(0)}\},\{(\syntax{p(0)}, \syntax{p(0)})\})$ and $G_n$ ($n=1,2,\ldots$) to be the hypothesis graph $(V^{(n)}, E^{(n)}_\syntax{R})$ 
such that $V^{(n)}$ is the union of:
\begin{gather*}
\{\syntax{p(0)}\},\quad \{\syntax{p(x}_i) \mid 0\lneq i\leq n) \},\\
\{\syntax{0} \neq \syntax{x}_i \mid 0\lneq  i \leq n\},\quad
 \{ \syntax{x}_i\neq \syntax{x}_j \mid 0\lneq  i \lneq  j \leq n \};
\end{gather*}
and $E^{(n)}_\syntax{R}$ is the union of:
\begin{gather*}
\{(\syntax{p(0)}, \syntax{p(x}_1)), (\syntax{p(x}_n), \syntax{p(0)})\} \text{ and } \{(\syntax{p(x}_i), \syntax{p(x}_{i+1}))\mid 0\lneq i\lneq n\}.
\end{gather*}
Then, we have $Sent(G_n)=H_n$ for $n=0,1,\ldots$.
Therefore, for $n,n' =0,1,\ldots$,  $G_n$ and $G_{n'}$ cannot be embedded into each other if $n\neq n'$ by Proposition \ref{prop: hom induces implication} since $H_n$ and $H_{n'}$ do not logically imply each other in $\class C$ if $n\neq n'$.
Hence, $\mc G_i=\{G_i\}$ and $\mc F_i =\{([G_i]_{\simeq},neg)\}$ for $i=1,2,\ldots$ constitute an infinite dialogue $\mc G_1, \mc F_1, \mc G_2, \mc F_2, \ldots$ in $\protocolTowards {\protocolProp {\protocolname{Basic}} {\propertyG O {\class C}} } {\{G_0\}}$  and thus in $\protocolProp {\basicUFBD} {\propertyG O {\class C}}$.\proved\checked
$\hfill\dashv$

\end{example}
\noindent Imposing the following condition, we obtain the halting property:
\begin{definition}\label{definition: n-bounded size constraint}
Let $n$ be a non-negative integer.
We say that a finite or an infinite dialogue $\mc G_1, \mc F_1,\dots, \mc G_i, \mc F_i, \ldots$ in $\protocolProp {\UFBD} {\propertyG O {\class C}}$ \emph{satisfies the $n$-bounded size constraint} if it satisfies the condition for each index $i$: there is an FB $(g, f)\in \mc F_i$ with $|g|\leq n$ if $\mc F_i$ is non-empty.

\end{definition}
\noindent For a UFBD protocol $\class D$ on $\propertyG O {\class C}$, \emph{the UFBD protocol $\class D$ with the $n$-bounded size constraint} (notation: $\protocolSizeRestrict {\class D} n $) is defined to be the subset of  $\class D$ whose dialogues satisfy the $n$-bounded size constraint.
It is easy to see that $\protocolTowards{ \protocolSizeRestrict {\class D} n }  \targetsetG=  \protocolSizeRestrict{ \protocolTowards{\class D} \targetsetG }  n$ for any non-empty $\targetsetG\subseteq \hypgraphset O {\class C}$ and non-negative integer $n$.\checked\checkedC

\begin{theorem}\label{DecidabilityOfGraphUFBD}\proved\checked\checkedC
The protocol $\protocolSizeRestrict {\protocolname{Basic}(\propertyG O {\class C})} n$ has the halting property for any non-negative integer $n$.
\end{theorem}
\begin{proof}
This is proved in a way similar to Theorem \ref{theorem: halting property of syntax-based dialogue protocol}.
Let $(\quotient {\formgraphset} {\simeq})_n$ denote $\{g\in \quotient {\formgraphset} {\simeq} \mid  \cardinal g \leq n \}$.
Then, it is easy to see that the cardinality $\cardinal { (\quotient {\formgraphset} {\simeq})_n }$ is finite since there are at most finitely many formula graphs whose orders are less than or equal to $n$ up to isomorphism by Proposition \ref{proposition: size n formula graphs finite}.
Suppose, for the sake of contradiction, that there is an infinite dialogue $\mc D= \mc G_1, \mc F_1, \ldots$ in $\protocolSizeRestrict {\protocolname{Basic}(\propertyG O {\class C})} n$.
Let also $S_i$ be the set $\{g\in (\quotient {\formgraphset} {\simeq})_n \mid (g, f)\in \mc F_i\}$ for each index $i$.
Then, $S_i$ for each index $i$ is non-empty since $\mc F_i$ is non-empty due to Condition \ref{condition: empty FB is the last} of Definition \ref{definition: basic UFBD}.
Let us write $S'_i$ to mean $\bigcup_{j=1}^i S_j$.
Since $S_j$ does not contain $g$  for any indices $i,j$ with $j\lneq i$ and $g\in S_i$ due to Condition \ref{condition: no previous FB} of Definition \ref{definition: basic UFBD}, we have $S'_i \subsetneq S'_j$. 
Using this fact, we can easily see that $\cardinal{S'_i}\geq i$ for any index $i$ by induction on $i$.
Let us now take an index $j\gneq \cardinal{ (\quotient {\formgraphset} {\simeq} )_n}$.
Then, we obtain $\cardinal{S'_j}\geq j\gneq \cardinal{(\quotient {\formgraphset} {\simeq})_n}$, which contradicts to $S'_j\subseteq (\quotient {\formgraphset} {\simeq})_n$.\checked\checkedC
\end{proof}

We conclude this section by proving that $\protocolSizeRestrict {\protocolname{Basic}(\propertyG O {\class C})} n$ has the convergence property too.
\begin{lemma}\label{lemma: next F w cond 2e}\proved\checked\checkedC
Let $\targetsetG$ be a non-empty subset of $\hypgraphset O {\class C}$, $n$ be a non-negative integer such that $n \gneq \cardinal{G}$ for any $G\in \targetsetG$ and  $\mc D=\mc G_1, \mc F_1,\dots, \mc G_i$ $(i\geq 1)$ be a dialogue in $\protocolTowards{ \protocolSizeRestrict {\protocolname{Basic}(\propertyG O {\class C})} n }  \targetsetG$.
Suppose that there is a set $\mc F_i$ of FBs such that $\mc D, \mc F_i$ is a dialogue in $\protocolTowards {\protocolname{Basic}(\propertyG O {\class C})} \targetsetG $.
Then, there is a set $\mc F'_i$ of FBs such that $\mc D, \mc F'_i$ is a dialogue in $\protocolTowards{ \protocolSizeRestrict {\protocolname{Basic}(\propertyG O {\class C})} n }  \targetsetG$.
\end{lemma}
\begin{proof}
If $\mc F_i$ is empty, then $\mc D, \mc F_i$ is a dialogue in $\protocolSizeRestrict {\protocolTowards {\protocolname{Basic}(\propertyG O {\class C})} \targetsetG } n$ thus in $\protocolTowards{ \protocolSizeRestrict {\protocolname{Basic}(\propertyG O {\class C})} n }  \targetsetG$.
Hence, suppose that $\mc F_i$ is non-empty.
If $\mc F_i$ contains a positive FB $(g, pos)$, we have $|g|\leq|G| \lneq n$ for any $G\in\targetsetG$ and thus  $\mc D, \mc F_i\in  \protocolTowards{ \protocolSizeRestrict {\protocolname{Basic}(\propertyG O {\class C})} n }  \targetsetG$ holds.
If $\mc F_i$ contains a neutral FB $(g, neutral)$, then $\cardinal g \leq \cardinal G \lneq n$ holds for some $G\in \targetsetG$ and thus we have $\mc D, \mc F_i\in  \protocolTowards{ \protocolSizeRestrict {\protocolname{Basic}(\propertyG O {\class C})} n }  \targetsetG$ again.
Therefore, we suppose $f=neg$ and $\cardinal g \gneq n$ for any $(g, f)\in\mc F_i$.
Let $([G']_\simeq, neg)$ be an FB in $\mc F_i$ with $\cardinal{[G']_\simeq}\gneq n$
 and $G''$ be a subgraph of $G'$ with $\max \{\cardinal{G}\mid G\in  \targetsetG\} \lneq \cardinal{G''}\leq n$.
 Such $G''$ exists: e.g. a subgraph $\hypograph{V}{E}$ such that $V$ is any subset of the set of vertices of $G'$ with $\cardinal V = n$ and that $E_\syntax{R}$ is empty for any $\syntax{R}\in\mc R$.
Define $\mc F'_i$ to be $\mc F_i\cup\{([G'']_\simeq,neg)\}$.
Then,
 $\mc D, \mc F'_i$ is a dialogue in $\protocolTowards{ \protocolSizeRestrict {\protocolname{Basic}(\propertyG O {\class C})} n }  \targetsetG$.
In particular, $[G'']_\simeq$ has not been pointed out before $i$ because otherwise $\mc D$ violates Condition \ref{condition: satisfy all previous FBs} of Definition \ref{definition: UFBD}.\checked\checkedC
\end{proof}

\begin{lemma}\label{lemma: next G w cond 2e}\proved\checked\checkedC
Let $\targetsetG$ be a non-empty subset of $\hypgraphset O {\class C}$, $n$ be a non-negative integer and $\mc D=\mc G_1, \mc F_1, \dots, \mc G_i, \mc F_i$ be a finite user-feedback dialogue in the protocol $\protocolTowards{ \protocolSizeRestrict {\protocolname{Basic}(\propertyG O {\class C})} n }  \targetsetG$.
If there is a set $\mc G_{i+1}\subseteq\hypgraphset O {\class C}$ such that  $\mc D, \mc G_{i+1}$ is a dialogue in $\protocolTowards {\protocolname{Basic}(\propertyG O {\class C})}  \targetsetG$, then 
	 $\mc D, \mc G_{i+1}$ is also a dialogue in $\protocolTowards{ \protocolSizeRestrict {\protocolname{Basic}(\propertyG O {\class C})} n }  \targetsetG$.
\end{lemma}
\begin{proof}
It is easy to see that this $\mc D, \mc G_{i+1}$ satisfies the $n$-bounded size constraint.
Hence, $\mc D, \mc G_{i+1}$ is a dialogue in   $\protocolSizeRestrict  { \protocolTowards {\protocolname{Basic}(\propertyG O {\class C})} \targetsetG } n $ and thus in $\protocolTowards{ \protocolSizeRestrict {\protocolname{Basic}(\propertyG O {\class C})} n }  \targetsetG$.\checked\checkedC
\end{proof}

\begin{theorem}\label{the last set with Condition 2e}\proved\checked\checkedC
Let $\targetsetG$ be a non-empty subset of $\hypgraphset O {\class C}$, $n$ be a non-negative integer.
Then, any finite dialogue in $\protocolTowards {\protocolSizeRestrict {\protocolname{Basic}(\propertyG O {\class C})} n} \targetsetG$ converges to $\targetsetG$ in $\protocolTowards {\protocolSizeRestrict {\protocolname{Basic}(\propertyG O {\class C})} n} \targetsetG$.
\end{theorem}
\begin{proof}
Let $\targetsetG$ be a non-empty subset of $\hypgraphset O {\class C}$ and $\mc D$ be a finite dialogue in $\protocolTowards{ \protocolSizeRestrict {\protocolname{Basic}(\propertyG O {\class C})} n }  \targetsetG$ such that there is no set $\mc A$ satisfying that the sequence $\mc D, \mc A$ is a dialogue in $\protocolTowards{ \protocolSizeRestrict {\protocolname{Basic}(\propertyG O {\class C})} n }  \targetsetG$.
Then, there is not a set $\mc A'$ such that $\mc D, \mc A'$ is a dialogue in $\protocolTowards {\protocolname{Basic}(\propertyG O {\class C})}  \targetsetG$ because otherwise
there is a set $\mc A''$ such that the sequence $\mc D, \mc A''$ is a dialogue in $\protocolTowards{ \protocolSizeRestrict {\protocolname{Basic}(\propertyG O {\class C})} n }  \targetsetG$ by Lemmata \ref{lemma: next F w cond 2e} and \ref{lemma: next G w cond 2e}.
Therefore, we can apply Theorem \ref{theorem: convergence} to see that Conditions \ref{condition: last set in convergence} and \ref{condition: satisfy shared properties} in Definition \ref{definition: convergence of one dialogue} are satisfied.\checked\checkedC
\end{proof}

\subsection{Simple UFBD Protocols}\label{section: simple UFBD}
In the previous section we have defined the concept of basic UFBD.
In this section, we first propose that of \emph{simple user-feedback dialogue} (\emph{simple UFBD}), which is a special case of a basic UFBD protocol as we shall see.
We then see that the simple UFBD protocol $\protocolProp \simpleUFBD \property$, or the set of simple UFBDs, for some property assignment function $\property$ does not have the convergence property while it does for any property assignment function the convergence property towards a single target.
We also see that $\protocolProp \simpleUFBD {\propertyG O {\class C}}$ and $\protocolProp \simpleUFBD {\propertyH O {\class C}}$ does not necessarily have the halting property while they do under the same conditions as Theorems \ref{theorem: halting property of syntax-based dialogue protocol} and  \ref{the last set with Condition 2e}.

\subsubsection{Generic Results on Simple UFBD Protocols}
We first define the concept of \emph{simple user-feedback dialogue} (\emph{simple UFBD}) and see generic results:
\begin{inparaenum}[(1)]
\item
a simple UFBD is a special case of a basic UFBD;
\item
the simple UFBD protocol, i.e. the set of simple UFBDs, has the convergence property towards single target.
\end{inparaenum}

The concept of simple user-feedback dialogue defined as below:
\begin{definition}\label{definition: simple UFBD}
Let $\mc X$ be a set of candidates, $P$ be a set of properties and $Prop:\mc X\to \wp(P)$ be a property assignment function.
A \emph{simple user-feedback dialogue} (\emph{simple UFBD}) $\mc D=\mc X_1, \mc F_1, \mc X_2, \mc F_2,\dots, \mc X_i, \mc F_i, \ldots$ on $\property$ is a finite or an infinite dialogue in $\protocolProp \UFBD \property$ that satisfies the conditions below for each index $i$.
\begin{enumerate}
\item Conditions on $\mc X_i$. 
\begin{enumerate}
\item \label{condition: distinct implies not equivalent simple UFBD}
For any $X,X'\in \mc X_i$ with $X\neq X'$, \camera{ \emph{$X$ and $X'$ have distinct sets of properties}, i.e} $Prop(X)\neq Prop(X')$.
\item \label{condition: multiple elements X simple UFBD}
 $\mc X_i$ contains multiple candidates if there are multiple candidates in $\mc X$ that \camera{have distinct sets of properties and that satisfy all  previous FBs}.
\item \label{condition: last X simple UFBD}
If $\mc X_i$ is a singleton or the empty set, it is the last set of $\mc D$.
\end{enumerate}
\item Conditions on $\mc F_i$.
\begin{enumerate}
\item \label{condition: pos or neg simple UFBD}
For any $(p, f)\in \mc F_i$,  $f$ is either  $pos$ or $neg$.
\item \label{condition: FB prop should appear simple UFBD}
For any $(p, f)\in \mc F_i$, $p$ appears in $\mc X_i$.
\item \label{condition: FB prop should be new simple UFBD}
For any $(p, f)\in \mc F_i$, $p$ has not been pointed out before $i$.

\item \label{condition: no-unpointed property in simple UFBD}
$\mc F_i$ is non-empty if there is a property $p\in P$ that appears in $\mc X_i$ and that has not been pointed out before $i$.
\item \label{condition: last set F in simple UFBD}
If $\mc F_i$ is empty, it is the last set of $\mc D$.
\end{enumerate}

\end{enumerate}
A \emph{simple UFBD protocol on $\property$} is defined to be a set of simple UFBDs on $\property$.
\emph{The simple UFBD protocol on $\property$} (notation: $\protocolProp \simpleUFBD \property$) is defined to be the set of simple UFBDs on $\property$.
\end{definition}

Let us first see that any simple UFBD is a basic UFBD, that is, for any property assignment function $\property$, $\protocolProp \simpleUFBD \property\subseteq \protocolProp \basicUFBD \property$ holds.
\begin{lemma}\label{lemma: distinct has un-pointed-out property}\proved\checked\checkedC
Let $\property: \mc X\to \powerset P$ be a property assignment function and $\mc D=\mc X_1,\mc F_1, \ldots, \mc X_i$ $(i\geq 1)$ be a finite dialogue in $\protocolProp \simpleUFBD \property$.
If  $|\mc X_i|$ is greater than or equal to $2$, then there are $X\in\mc X_i$ and $p\in \property(X)$ such that $p$ has not been pointed out before $i$.
\end{lemma}
\begin{proof}
Take arbitrary distinct two candidates $X_1$ and $X_2$ in $\mc X_i$.
Then, we have $Prop(X_1)\neq Prop(X_2)$ by Condition \ref{condition: distinct implies not equivalent simple UFBD} of Definition \ref{definition: simple UFBD}. 
Therefore, there is a property $p\in P$ such that $p\in Prop(X_1)\setminus Prop(X_2)$ or $p\in {Prop(X_2)}\setminus{Prop(X_1)}$.
We show that this $p$ has not been pointed out before $i$.
We may suppose $p\in Prop(X_1)\setminus Prop(X_2)$ without loss of generality.
We also suppose, for the sake of contradiction, that $p$ has been pointed out before $i$.
Then, if  $(p, pos)\in \mc F_j$ for some $j\lneq i$, then this contradicts $p\notin Prop(X_2)$; if $(p, neg)\in \mc F_j$ for some $j\lneq i$, this contradicts $p\in Prop(X_1)$.
Therefore, $p$ has not been pointed out before $i$.\proved\checked\checkedC
\end{proof}

\begin{proposition}\label{proposition: simple UFBD is UFBD}
It holds that $\protocolProp \simpleUFBD \property\subseteq \protocolProp \basicUFBD \property$ for any property assignment function $\property$.\proved\checked\checkedC
\end{proposition}
\begin{proof}
We show that any dialogue in $\protocolProp \simpleUFBD \property$ satisfies Conditions \ref{condition: un-poited-out property} and \ref{condition: last X} of Definition \ref{definition: basic UFBD}; the other conditions of Definition \ref{definition: basic UFBD} are included in Definition \ref{definition: simple UFBD}.
Let $\mc D=\mc X_1, \mc F_1,\dots, \mc X_i, \mc F_i, \ldots$ be a finite or an infinite dialogue in $\protocolProp \simpleUFBD \property$ and $i$ be an arbitrary index of $\mc D$.

We first see that $\mc X_i$ satisfies Condition \ref{condition: un-poited-out property}.
If $\mc X_i$ is empty, then, there is no candidate in $\mc X$ that satisfies all previous FBs by Condition \ref{condition: non-empty X} of Definition \ref{definition: UFBD}.
 Therefore, $\mc X_i$ satisfies Conditions \ref{condition: un-poited-out property}.
If $\mc X_i$ is a singleton, it holds that $\property(X') = \property(X)$ for any two candidates $X, X'\in \mc X$ satisfying all previous FBs, due to Condition \ref{condition: multiple elements X simple UFBD} of Definition \ref{definition: simple UFBD}.
Therefore, if the only candidate in $\mc X_i$ has a property $p$ that has not been pointed out before $i$, then $\mc X_i$ satisfies Condition \ref{condition: un-poited-out property} of Definition \ref{definition: basic UFBD}.
Otherwise, there is no $X\in\mc X$ such that some $p\in \property(X)$ that has not been pointed out before $i$. Hence, $\mc X_i$ satisfies Condition \ref{condition: un-poited-out property} of Definition \ref{definition: basic UFBD}.
 In the case where $\cardinal{\mc X_i}\geq 2$, there are $X\in \mc X_i$ and $p\in\property(X)$ such that $p$ has not been pointed out before $i$ by Lemma \ref{lemma: distinct has un-pointed-out property} and thus $\mc X_i$ satisfies Condition \ref{condition: un-poited-out property}.

We next show that $\mc X_i$ satisfies Condition \ref{condition: last X} of Definition \ref{definition: basic UFBD}.
Suppose that $p$ has been pointed out before $i$ for any $X\in\mc X_i$ and $p\in \property(X)$.
Then, by the contraposition of Lemma \ref{lemma: distinct has un-pointed-out property}, we have $\cardinal{X_i}\leq 1$.
Therefore, $\mc X_i$ is the last set by Condition \ref{condition: last X simple UFBD} of Definition \ref{definition: simple UFBD}.
\proved\checked\checkedC
\end{proof}

The simple UFBD protocols $\protocolProp \simpleUFBD {\propertyH O {\class C}}$ does not necessarily have the convergence property. In addition,  $\protocolTowards { \protocolSizeRestrict  { \protocolProp \simpleUFBD {\propertyG O {\class C}}} n} \targetsetG$ does not necessarily converges $\targetsetG$ as illustrated by the following example:
\begin{example}\label{example: target set in simple UFBD}
Consider a single-sorted language constructed from  an alphabet $\Sigma=(\mathcal{S}, \mathcal{C}, \mathcal{P}, \mathcal{R},\mathcal{V}_1,\mathcal{V}_2)$ 
such that $\constsetC=\{\syntax{0}\}$, $\propsetP=\{\syntax{p_1}, \syntax{p_2}, \syntax{p_3}\}$ and
$\relsetR=\emptyset$ and
the class $\class C=\{\structure M I\}$ of the structure $\structure M I$ such that $M=\{0\}$,  $\syntax{0}^{\mc M}=0$, $\syntax{p_1}^{\mc M}=\syntax{p_2}^{\mc M}=\syntax{p_3}^{\mc M}=\{0\}$.
Suppose that the set of observations $O$ is empty and that the target set $\mc G_t$ is $\{G_1, G_2\}$, where $G_1=(\{\syntax{p_1(0)}, \syntax{p_3(0)}\},\fami{E_{\syntax R}} {\syntax{R}} \relsetR)$ and $G_2=(\{\syntax{p_2(0)}, \syntax{p_3(0)}\},\fami{E_{\syntax R}} {\syntax{R}} \relsetR)$.
Note that $\fami {E_{\syntax R}} {\syntax{R}} \relsetR$ is an $\emptyset$-indexed family and thus uniquely determined.
Let $\mc G_1$ be $\{G_3, G_4\}$, where $G_3=(\{\syntax{p_1(0)}\},\fami{E_{\syntax R}} {\syntax{R}} \relsetR)$ and $G_4=(\{\syntax{p_2(0)}\}, \fami{E_{\syntax R}} {\syntax{R}} \relsetR)$, and $\mc F_1$ be $\{([(\emptyset,\fami{E_{\syntax R}} {\syntax{R}} \relsetR)]_{\simeq}, pos)\}$.
Then, it is easy to see that: (respectively, for any non-negative integer $n$)
\begin{enumerate}
\item
$\mc G_1, \mc F_1, \mc G_1$ is a dialogue in the protocol $\protocolTowards  { \protocolProp \simpleUFBD {\propertyG O {\class C}}} {\mc G_t}$ (respectively, $\protocolTowards { \protocolSizeRestrict  { \protocolProp \simpleUFBD {\propertyG O {\class C}}} n} {\mc G_t}$);

\item
there is no set $\mc F_2$ of FBs such that $\mc G_1, \mc F_1, \mc G_1, \mc F_2$ is a dialogue in $\protocolTowards  { \protocolProp \simpleUFBD {\propertyG O {\class C}}} {\mc G_t}$  (respectively, $\protocolTowards { \protocolSizeRestrict  { \protocolProp \simpleUFBD {\propertyG O {\class C}}} n} {\mc G_t}$);

\item
both of the following hold:
\begin{enumerate}
\item
$[( \{\syntax{p_3(0)} \}, \fami{E_{\syntax R}} {\syntax{R}}  \relsetR )]_{\simeq} \in \propertyG  O {\class C}(G)$ for any $G\in \mc G_t$ and 
\item
$[( \{\syntax{p_3(0)} \}, \fami{E_{\syntax R}} {\syntax{R}}  \relsetR )]_{\simeq} \notin \propertyG  O {\class C}(G')$ for any $G'\in\mc G_1\neq \emptyset$.
\end{enumerate}
\end{enumerate}
Therefore, $\protocolProp \simpleUFBD {\propertyG O {\class C}}$ does not have the convergence property (respectively, $\protocolTowards { \protocolSizeRestrict  { \protocolProp \simpleUFBD {\propertyG O {\class C}}} n} {\mc G_t}$) does not converges to $\targetsetG$). \proved\checked \hfill $\dashv$
\end{example}
\noindent The simple UFBD protocol $\protocolTowards  { \protocolProp \simpleUFBD {\propertyH O {\class C}}} {\mc H_t}$ does not necessarily have the convergence property either:
\begin{example}\label{example: simple UFBD not convergence propH to set}
Consider a single-sorted language such that $\mc C=\{\syntax 0\}$, $\mc P=\{\syntax{p}_1, \syntax{p}_2\}$ and $\mc R=\emptyset$.
Let $\class C$ be the set $\{(\mc{M}_S, \mc I_S) \mid S\in\powerset {\{1,2\} }\}$ such that the carrier set of $\mc M_S$ is $\{0\}$, ${\syntax 0}^{\mc M_S}=\{0\}$ and ${\syntax p_i}^{\mc M_S}$  ($i=1,2$)  is $\{0\}$ if $i\in S$ and $\emptyset$ otherwise.
Suppose that $O$ be the empty set of sentences and $\targetsetH=\{H_1, H_2\}$ be the target set of hypotheses, where  $H_1=\{\syntax{p_1(0)\land \neg p_2(0)}\}$ and $H_2=\{\syntax{\neg p_1(0) \land p_2(0)}\}$.
Let $\mc H_1$ be $\{H_3, H_4\}$, where $H_3=\{\syntax{p_1(0)}\}$ and $H_4=\{\syntax{p_2(0)}\}$, and $\mc F_1=\{([\syntax{\top}]_{\leftrightarrow_{\class C}}, pos), ([\syntax{p_1(0)\lor p_2(0)}]_{\leftrightarrow_{\class C}}, pos)\}$.
Then, the following hold:
\begin{enumerate}
\item \label{condition: example not convergence propH make a dialogue}
$\mc H_1, \mc F_1, \mc H_1$ is a dialogue in $\protocolTowards  { \protocolProp \simpleUFBD {\propertyH O {\class C}}} {\mc H_t}$;

\item \label{condition: example not convergence propH termination}
there is no set $\mc F_2$ of FBs such that $\mc H_1, \mc F_1, \mc H_1, \mc F_2$ is a dialogue in the protocol $\protocolTowards  { \protocolProp \simpleUFBD {\propertyH O {\class C}}} {\mc H_t}$;

\item \label{condition: example not convergence propH not convergence}
both of the following hold:
\begin{enumerate}
\item 
$[\syntax{\neg(p_1(0)\land p_2(0))}]_{\leftrightarrow_{\class C}} \in \propertyH  O {\class C}(H)$ for any $H\in \targetsetH$ and 
\item
$[\syntax{\neg(p_1(0)\land p_2(0))}]_{\leftrightarrow_{\class C}}  \notin \propertyH  O {\class C}(H')$ for any $H'\in \mc H_1$.
\end{enumerate}
\end{enumerate}
This means that the protocol $\protocolProp \simpleUFBD {\propertyH O {\class C}}$ does not have the convergence property.
We see Condition \ref{condition: example not convergence propH termination} only; the other conditions are easy to prove.
Define $f:\sentenceSet \to \powerset {\{1,2\}} $ by $f(\syntax{\Phi}) = \{ S\in \powerset {\{1,2\}} \mid (\mc M_S, \mc I_S)\models \syntax{\Phi}  \}$ for any $\syntax{\Phi}\in\sentenceSet$.
Note that $f(\syntax{\Phi})=f(\syntax{\Psi})\Leftrightarrow [\syntax{\Phi}]_{\leftrightarrow_{\class C}} = [\syntax{\Psi}]_{\leftrightarrow_{\class C}}$ holds for any $\syntax{\Phi}, \syntax{\Psi}\in\sentenceSet$.
We also define $Sup: \powerset {\{1,2\}}
 \to \powerset {\powerset{\{1,2\}}}$ by $Sup(S)=\{S' \in \powerset {\powerset{\{1,2\}}}\mid S\subseteq S' \}$.
Then, it is easy to see that:
\[\syntax{\Phi}\models_{\class C}\syntax{\Psi} \Leftrightarrow f(\syntax{\Phi})\subseteq f(\syntax{\Psi)}\Leftrightarrow  f(\syntax{\Psi}) \in Sup(f(\syntax{\Phi})),\]
 for any $\syntax{\Phi}, \syntax{\Psi}\in\sentenceSet$.
Suppose, for the sake of contradiction, that there is $\mc F_2$ such that $\mc H_1, \mc F_1, \mc H_1, \mc F_2$ is a dialogue $\protocolTowards  { \protocolProp \simpleUFBD {\propertyH O {\class C}}} {\mc H_t}$.
Then, $\mc F_2$ is non-empty due to Condition \ref{condition: no-unpointed property in simple UFBD} in \ref{definition: simple UFBD} and  Lemma \ref{lemma: distinct has un-pointed-out property}.
If $\mc F_2$ contains a positive FB $([\syntax{\Phi}], pos)$, then we have:
\[f(\syntax{\Phi}) \in Sup(f(\syntax{p_1(0)\land \neg p_2(0)})) \cap  Sup(f(\syntax{\neg p_1(0)\land p_2(0)})),\]
 by the positivity of the FB and that 
$f(\syntax{\Phi}) \in  Sup(f(\syntax{p_1(0)}))\cup (Sup(f(\syntax{p_2})))$ by Condition \ref{condition: FB prop should appear simple UFBD} in Definition \ref{definition: simple UFBD}.
Then, it is easy to see that $f(\syntax{\Phi})\ = f(\top)$ or $ f(\syntax{\Phi})= f(\syntax{p_1(0)\lor p_2(0)})$. Therefore, $[\syntax{\Phi}]_{\leftrightarrow_{\class C}} = [\top]_{\leftrightarrow_{\class C}}$ or $[\syntax{\Phi}]_{\leftrightarrow_{\class C}} = [\syntax{p_1(0)\lor p_2(0)}]_{\leftrightarrow_{\class C}}$. This contradicts to Condition \ref{condition: FB prop should be new simple UFBD} in Definition \ref{definition: simple UFBD}.
We can see, in an analogous way, that there is no negative FB in $\mc F_2$.
Therefore, there is no $\mc F_2$ such that $\mc H_1, \mc F_1, \mc H_1, \mc F_2$ is a dialogue $\protocolTowards  { \protocolProp \simpleUFBD {\propertyH O {\class C}}} {\mc H_t}$.
\proved\checked \hfill $\dashv$
\end{example}

We next see that 
the protocol $\protocolProp \simpleUFBD \property$ has the convergence property towards a single target for any property assignment function $\property$.
\begin{lemma}\label{lemma: next F in simple UFBD}\proved\checked\checkedC
Let $\property:\mc X\to \powerset P$ be a property assignment function, $\targetX$ be a candidate in $\mc X$ and $\mc D=\mc X_1,\mc F_1, \ldots, \mc X_i$ $(i\geq 1)$ be a finite dialogue in $\protocolTowards {\protocolProp \simpleUFBD \property} {\{\targetX\}}$.
If $\mc X_i$ contains multiple distinct candidates, then there is a non-empty set $\mc F_i$ of FBs on $P$ such that the sequence $\mc D, \mc F_i$ is again a dialogue in $\protocolTowards {\protocolProp \simpleUFBD \property} {\{\targetX\}}$.
\end{lemma}
\begin{proof}
It is easy to see that $\mc D$ is a dialogue in $\protocolTowards {\protocolProp \basicUFBD \property} {\{X_t\}}$ by Proposition \ref{proposition: simple UFBD is UFBD}.
Therefore, by Lemmata \ref{lemma: next FB} and \ref{lemma: distinct has un-pointed-out property},  
if $\mc X_i$ contains multiple candidates, there is a non-empty set $\mc F_i$ of FBs on $P$ such that the sequence $\mc D, \mc F_i$ is again a $\protocolTowards {\protocolProp \basicUFBD \property} {\{X_t\}}$.
This sequence  $\mc D, \mc F_i$ is also a dialogue in $\protocolTowards {\protocolProp \simpleUFBD \property} {\{X_t\}}$ by Proposition \ref{proposition: single target implies no neutral}.\checked\checkedC
\end{proof}

\begin{lemma}\label{lemma: next X in simple UFBD}\proved\checked\checkedC
Let $\property:\mc X\to \powerset P$ be a property assignment function, $\targetX$ be a candidate in $\mc X$ and $\mc D=\mc X_1, \mc F_1, \ldots, \mc X_i, \mc F_i$ be a finite dialogue in $\protocolTowards {\protocolProp \simpleUFBD \property} {\{X_t\}}$.
Then, there exists a set $\mc X_{i+1}\subseteq \mc X$ such that $\mc D, \mc X_{i+1}$ is a dialogue in $\protocolTowards {\protocolProp \simpleUFBD \property} {\{X_t\}}$.
\end{lemma}
\begin{proof}
This is proved in a way similar to the proof of Lemma \ref{lemma: exist next X}.
We first see that Condition \ref{condition: last set F in simple UFBD} of Definition \ref{definition: simple UFBD} does not apply to $\mc D$.
This is trivial by definition if $\mc D$ is the null sequence.
Hence, we assume that $\mc D$ is not null and show that $\mc F_i$ is non-empty.
We can see that $\mc X_i$ contains multiple candidates by Condition \ref{condition: last X simple UFBD} of Definition \ref{definition: simple UFBD} since $\mc X_i$ is not the last set of $\mc D$.
Hence, there are $X\in \mc X_i$ and $p\in\property(X)$ such that $p$ has not been pointed out before $i$ by Lemma \ref{lemma: distinct has un-pointed-out property}.
Therefore, we see that that $\mc F_i$ is non-empty by Condition \ref{condition: no-unpointed property in simple UFBD} of of Definition \ref{definition: simple UFBD}.
We then define $\mc X_{i+1}$ as follows:
it is defined to be the set of two candidates in $\mc X$ that have distinct sets of properties and that satisfy all  previous FBs if such candidates exist;
it is defined to be $\{\targetX\}$ otherwise.
Then, the resulting sequence $\mc D, \mc X_{i+1}$ is a dialogue in $\protocolTowards {\protocolProp \simpleUFBD \property} {\{X_t\}}$.\proved\checked\checkedC
\end{proof}

\begin{theorem}\label{theorem: simple convergence towards single target}\proved\checked\checkedC
The protocol $\protocolProp \simpleUFBD \property$ has the convergence property towards a single target for any property assignment function $\property$.
\end{theorem}
\begin{proof}
This is proved in a way similar to the proof of Theorem \ref{theorem: convergence}.
Let $\property:\mc X\to \powerset P$ be a property assignment function, $\targetX$ be a candidate in $\mc X$ and  
$\mc D$ be a finite dialogue in $\protocolTowards {\protocolProp \simpleUFBD \property} {\{X_t\}}$.
The last set of $\mc D$ is a set of candidates $\mc X$ by Lemma \ref{lemma: next X in simple UFBD} and it does not contain multiple candidates by Lemma \ref{lemma: next F in simple UFBD}.
It is also non-empty by Condition \ref{condition: non-empty X} of Definition \ref{definition: UFBD} since $X_t$ satisfies all previous FBs.
Thus, it is a singleton $\{X\}\subseteq \mc X$.
We can also see that $\property(X)=\property(X_t)$ by Condition \ref{condition: multiple elements X simple UFBD} of Definition \ref{definition: simple UFBD} since $X$ and $X_t$ satisfies all previous FBs by Condition \ref{condition: satisfy all previous FBs} of Definition \ref{definition: UFBD} and by Definition \ref{definition: towards}, respectively.
Therefore, Condition \ref{condition: satisfy shared properties} of Definition \ref{definition: convergence of one dialogue} is satisfied.\proved\checked\checkedC
\end{proof}
\noindent As an immediate consequence of Theorem \ref{theorem: simple convergence towards single target}, we have:
\begin{corollary}\label{corollary: simple UFBD convergence towards single target}\proved\checked\checkedC
Let $\property:\mc X\to \powerset P$ be a property assignment function, $\targetX$ be a candidate in $\mc X$ and  
$\mc D$ be a finite dialogue in the protocol $\protocolTowards {\protocolProp \simpleUFBD \property} {\{X_t\}}$.
Suppose that there is no set $\mc A$ satisfying that $\mc D, \mc A$ is a dialogue in $\protocolTowards {\protocolProp \simpleUFBD \property} {\{X_t\}}$.
Then, the last set of $\mc D$ if the singleton of a candidate $X\in\mc X$ that satisfies $\property(X)=\property(X_t)$.\proved\checked\checkedC
\end{corollary}

\subsubsection{Simple UFBD Protocols on Hypotheses}
In the previous section, we have seen that $\protocolProp \simpleUFBD \property$ has the convergence property towards a single target for any property assignment function $\property$.
In the rest of this section, we see that $\protocolProp \simpleUFBD {\propertyH O {\class C}}$ and $\protocolProp \simpleUFBD {\propertyG O {\class C}}$ have the halting property under certain reasonable conditions.
We first see the case of $\protocolProp \simpleUFBD {\propertyH O {\class C}}$ and then see the case of $\protocolProp \simpleUFBD {\propertyG O {\class C}}$.

As the following example illustrates, $\protocolProp \simpleUFBD {\propertyH O {\class C}}$ does not necessarily have the halting property:
\begin{example}\label{example: no-halting property Simple propH}
We use the settings of Example \ref{example: semantics-based protocol not terminate}.
Let $\mc H'_i$ be $\{H_{2i-1}, H_{2i}\}$ and  $\mc F'_i$ be $\{([H_{2i-1}]_{\leftrightarrow_{\class C}},neg), ([H_{2i}]_{\leftrightarrow_{\class C}},neg)\}$ for each $i=1,2,\ldots$.
Then, it is easy to see that the sequence $\mc H'_1, \mc F'_1, \mc H'_2, \mc F'_2, \ldots$ is an infinite dialogue in $\protocolTowards {\protocolProp \simpleUFBD {\propertyH O {\class C} } } {\{H_0\}}$ and thus in $\protocolProp \simpleUFBD {\propertyH O {\class C} } $. \proved\checked\hfill $\dashv$
\end{example}

The protocol $\protocolProp \simpleUFBD {\propertyH O {\class C}}$ has the halting property under the same restriction of $\class C$ as in Theorem \ref{theorem: halting property of syntax-based dialogue protocol}, since $\protocolProp \simpleUFBD {\propertyH O {\class C}}\subseteq \protocolProp \basicUFBD {\propertyH O {\class C}}$ holds by Proposition \ref{proposition: simple UFBD is UFBD}:
\begin{corollary}\label{corollary: simple hypotheses halting}\proved\checked\checkedC
Let $\class{C}$ be a class of structures $\structure{M}{I}$ and $O$ a set of sentences. 
Suppose that the shared domain $M_\sigma$ is finite for any sort $\sigma\in\mc S$.
Then, $\protocolProp \simpleUFBD {\propertyH O {\class C}}$ has the halting property.
Hence, so does the protocol $\protocolTowards {\protocolProp \simpleUFBD {\propertyH O {\class C} } } {\targetsetH}$ for any set $\targetsetH\subseteq \propertyH O {\class C}$ of hypotheses.
\end{corollary}
\begin{proof}
This is an immediate consequence of  Theorem \ref{theorem: halting property of syntax-based dialogue protocol} and Proposition \ref{proposition: simple UFBD is UFBD}.\proved\checked\checkedC
\end{proof}
\noindent

The convergence property towards a single target (Corollary \ref{corollary: simple UFBD convergence towards single target}) for the protocol $\protocolProp \simpleUFBD {\propertyH O {\class C}}$ is rephrased as follows:
\begin{corollary}\label{corollary: convergence hyp}\proved\checked\checkedC
Let $\class{C}$ be a class of structures, $O$ a set of sentences $\targetH$ be a hypothesis in $\propertyH O {\class C}$ and $\mc D$ be a finite dialogue in the protocol $\protocolTowards {\protocolProp \simpleUFBD {\propertyH O {\class C} } } {\targetH}$.
Suppose that there is no set $\mc A$ satisfying that $\mc D, \mc A$ is a dialogue in $\protocolTowards {\protocolProp \simpleUFBD {\propertyH O {\class C} } } {\targetH}$.
Then, the last set of $\mc D$ is a singleton $\{H\}\subseteq \hypset O {\class C}$ satisfying $H \leftrightarrow_{\class C} \targetH$.
\end{corollary} 
\begin{proof}
This is immediately proved by Proposition \ref{prop: proph and equiv} and Corollary \ref{corollary: simple UFBD convergence towards single target}.\proved\checked\checkedC
\end{proof}

\subsubsection{Simple UFBD Protocols on Hypothesis Graphs}
We next see the case of $\protocolProp \simpleUFBD {\propertyG O {\class C}}$.
As illustrated by the following example, the protocol $\protocolProp \simpleUFBD {\propertyG O {\class C}}$ does not necessarily have the halting property:
\begin{example}\label{example: infinite syntax-based dialogue simple}
We use the settings of Example \ref{example: infinite syntax-based dialogue}.
Then, $\mc G'_i=\{G_{2i-1}, G_{2i}\}$ and $\mc F'_i =\{([G_{2i-1}]_{\simeq},neg), ([G_{2i}]_{\simeq},neg)\}$ for $i=1,2,\ldots$ constitute an infinite dialogue $\mc G'_1, \mc F'_1, \mc G'_2, \mc F'_2, \ldots$ in the protocol $\protocolTowards {\protocolProp \simpleUFBD {\propertyG O {\class C} } } {\{G_0\}}$ and thus in $\protocolProp \simpleUFBD {\propertyG O {\class C} }$. \proved\checked
$\hfill\dashv$

\end{example}

\begin{corollary}\label{corollary: simple graph n-bounded halting}\proved\checked\checkedC
Let $O$ be a set of sentences, $\class C$ be a class of structures and $n$ be a non-negative integer.
Then, $\protocolSizeRestrict  {\protocolProp \simpleUFBD {\propertyG O {\class C} }} n$ has the halting property.
\end{corollary}
\begin{proof}
It is easily seen that $\protocolSizeRestrict  {\protocolProp \simpleUFBD {\propertyG O {\class C} }} n \subseteq \protocolSizeRestrict  {\protocolProp \basicUFBD {\propertyG O {\class C} }} n $ using Proposition \ref{proposition: simple UFBD is UFBD}.
Since $\protocolSizeRestrict  {\protocolProp \basicUFBD {\propertyG O {\class C} }} n $ does not contain any infinite dialogue by Theorem \ref{DecidabilityOfGraphUFBD}, protocol $\protocolSizeRestrict  {\protocolProp \simpleUFBD {\propertyG O {\class C} }} n $ does not contains any infinite dialogue.
This means that $\protocolSizeRestrict  {\protocolProp \simpleUFBD {\propertyG O {\class C} }} n $ has the halting property.\proved\checked\checkedC
\end{proof}

We conclude this section by proving that $\protocolTowards {\protocolSizeRestrict  {\protocolProp \simpleUFBD {\propertyG O {\class C} }} n} {\{G_t\}}$ necessarily converges to $\{G_t\}$ if $n$ is strictly greater than the order $\cardinal {G_t}$ of $G_t$.
\begin{lemma}\label{lemma: simple n-bounded next F}\proved\checked\checkedC
Let $O$ be a set of sentences, $\class C$ be a class of structures, $G_t$ be a hypothesis graph in $\hypgraphset O {\class C}$, $n$ be a non-negative integer such that $n\gneq \cardinal {G_t}$  and  $\mc D=\mc G_1, \mc F_1,\dots, \mc G_i$ $(i\geq 1)$ be a finite dialogue in the protocol $\protocolTowards {\protocolSizeRestrict  {\protocolProp \simpleUFBD {\propertyG O {\class C} }} n} {\{G_t\}}$.
If $\mc G_i$ contains multiple hypothesis graphs, then there is a non-empty set $\mc F_i$ of FBs such that $\mc D, \mc F_i$ is again a dialogue in $\protocolTowards {\protocolSizeRestrict  {\protocolProp \simpleUFBD {\propertyG O {\class C} }} n} {\{G_t\}}$.
\end{lemma}
\begin{proof}
Since $\mc D$ is a dialogue in $\protocolTowards {\protocolSizeRestrict  {\protocolProp \simpleUFBD {\propertyG O {\class C} }} n} {\{G_t\}}$, it is also a dialogue in $\protocolTowards {\protocolProp \simpleUFBD {\propertyG O {\class C} }}  {\{G_t\}}$.
Therefore, by Lemma \ref{lemma: next F in simple UFBD} and the fact of $\cardinal{\mc G_i}\geq 2$, there is a non-empty set $\mc F_i$ of FBs such that the sequence $\mc D, \mc F_i $ is a dialogue in $\protocolTowards {\protocolProp \simpleUFBD {\propertyG O {\class C} }}  {\{G_t\}}$.
This sequence $\mc D, \mc F_i$ is also a dialogue in $\protocolTowards {\protocolProp \basicUFBD {\propertyG O {\class C} }}  {\{G_t\}}$ due to the fact that $\protocolTowards {\protocolProp \simpleUFBD {\propertyG O {\class C} }}  {\{G_t\}}\subseteq \protocolTowards {\protocolProp \basicUFBD {\propertyG O {\class C} }}  {\{G_t\}}$, which can be easily seen using Proposition \ref{proposition: simple UFBD is UFBD}.
Since $\mc D$ is also a dialogue in the protocol $\protocolTowards {\protocolSizeRestrict  {\protocolProp \basicUFBD {\propertyG O {\class C} }} n} {\{G_t\}}$,
 there is a non-empty set $\mc F'_i$ of FBs such that $\mc D, \mc F'_i$ is a dialogue in the protocol $\protocolTowards {\protocolSizeRestrict  {\protocolProp \basicUFBD {\propertyG O {\class C} }} n} {\{G_t\}}$  by Lemma \ref{lemma: next F w cond 2e}.
Therefore,
we can easily see that $\mc D, \mc F'_i\in \protocolTowards {\protocolSizeRestrict  {\protocolProp \simpleUFBD {\propertyG O {\class C} }} n} {\{G_t\}}$,
from the facts that $\mc D \in \protocolTowards {\protocolSizeRestrict  {\protocolProp \simpleUFBD {\propertyG O {\class C} }} n} {\{G_t\}}$ and that $\mc D, \mc F'_i\in \protocolTowards {\protocolSizeRestrict  {\protocolProp \basicUFBD {\propertyG O {\class C} }} n} {\{G_t\}}$.
\proved\checked\checkedC
\end{proof}

\begin{lemma}\label{lemma: simple n-bounded next G}\proved\checked\checkedC
Let $O$ be a set of sentences, $\class C$ be a class of structures, $G_t$ be a hypothesis graph in $\hypgraphset O {\class C}$, $n$ be a non-negative integer with $n\gneq \cardinal {G_t}$  and  $\mc D=\mc G_1, \mc F_1, \dots, \mc G_i, \mc F_i$ be a (possibly null) dialogue in the protocol $\protocolTowards {\protocolSizeRestrict  {\protocolProp \simpleUFBD {\propertyG O {\class C} }} n} {\{G_t\}}$.
	Then, there is a set $\mc G_{i+1}\subseteq\hypgraphset O {\class C}$ such that $\mc D, \mc G_{i+1}$ is a dialogue in $\protocolTowards {\protocolSizeRestrict  {\protocolProp \simpleUFBD {\propertyG O {\class C} }} n} {\{G_t\}}$.
\end{lemma}
\begin{proof}
Since $\mc D$ is a dialogue in $\protocolTowards {\protocolSizeRestrict  {\protocolProp \simpleUFBD {\propertyG O {\class C} }} n} {\{G_t\}}$, it is also a dialogue in $\protocolTowards  {\protocolProp \simpleUFBD {\propertyG O {\class C} }} {\{G_t\}}$.
Hence, by Lemma \ref{lemma: next X in simple UFBD}, there is a set $\mc G_{i+1}\subseteq\hypgraphset O {\class C}$ such that $\mc D, \mc G_{i+1}$ is a dialogue in the protocol $\protocolTowards  {\protocolProp \simpleUFBD {\propertyG O {\class C} }} {\{G_t\}}$.
Since the $n$-bounded size constraint does not impose any condition on  any subsets of hypothesis graphs in a UFBD on $\propertyG O {\class C}$,
we see that the sequence $\mc D, \mc G_{i+1}$ is a dialogue in the protocol $\protocolTowards {\protocolSizeRestrict  {\protocolProp \simpleUFBD {\propertyG O {\class C} }} n} {\{G_t\}}$.\proved\checked\checkedC
\end{proof}

\begin{theorem}\label{theorem: pointwise convergence of simple n-bounded on graph}\proved\checked\checkedC
Let $O$ be a set of sentences and $\class C$ be a set of structures.
For any hypothesis graph  $G_t \in \hypgraphset O {\class C}$ and non-negative integer $n$ with  $n\gneq \cardinal {G_t}$,
the protocol $\protocolTowards {\protocolSizeRestrict  {\protocolProp \simpleUFBD {\propertyG O {\class C} }} n} {\{G_t\}}$ converges to $\{G_t\}$.
\end{theorem}
\begin{proof}
This is proved in a way similar to the proof of Theorem \ref{theorem: simple convergence towards single target}.
Let $\mc D$ be a dialogue in $\protocolTowards {\protocolSizeRestrict  {\protocolProp \simpleUFBD {\propertyG O {\class C} }} n} {\{G_t\}}$ such that there is no set $\mc A$ such that $\mc D, \mc A$ is a dialogue in $\protocolTowards {\protocolSizeRestrict  {\protocolProp \simpleUFBD {\propertyG O {\class C} }} n} {\{G_t\}}$.
Then, we can see that the last set of $\mc D$ is a singleton $\{G\}\subseteq\hypgraphset O {\class C}$ by Condition \ref{condition: non-empty X} of Definition \ref{definition: UFBD} and Lemmata \ref{lemma: simple n-bounded next F} and \ref{lemma: simple n-bounded next G}.
We can also see that $PropG(O, {\class C})(G)=PropG(O, {\class C})(G_t)$ by Condition \ref{condition: multiple elements X simple UFBD} of Definition \ref{definition: simple UFBD} since both $G$ and $G_t$ satisfies all previous FBs by Condition \ref{condition: satisfy all previous FBs} of Definition \ref{definition: UFBD} and by Definition \ref{definition: towards}, respectively.\proved\checked\checkedC
\end{proof}

\begin{corollary}\proved\checked\checkedC
Let $O$ be a set of sentences, $\class C$ be a set of structures, $\targetG$ be a hypothesis graph in $\hypgraphset O {\class C}$, $n$ be a non-negative integer such that $n\gneq \cardinal \targetG$  and $\mc D$ be a finite dialogue in  $\protocolTowards {\protocolSizeRestrict  {\protocolProp \simpleUFBD {\propertyG O {\class C} }} n} {\{G_t\}}$.
Suppose that there is not a set $\mc A$ such that $\mc D, \mc A$ is a dialogue $\protocolTowards {\protocolSizeRestrict  {\protocolProp \simpleUFBD {\propertyG O {\class C} }} n} {\{G_t\}}$.
Then, the last set of $\mc D$ is the singleton of a hypothesis graph isomorphic to $\targetG$.

\end{corollary}
\begin{proof}
This is an immediate consequence of Theorem \ref{theorem: pointwise convergence of simple n-bounded on graph} and Proposition \ref{prop: prop and isom}.\proved\checked\checkedC
\end{proof}

\subsection{Ablation Studies: UFBD Protocol Types between $\simpleUFBD$ and $\basicUFBD$}\label{section: other UFBD protocols}
In the previous section, we replaced two kinds of conditions to obtain $\protocolProp \simpleUFBD \property$ from $\protocolProp \basicUFBD \property$: one is conditions on $\mc X_i$ and the other is ones on $\mc F_i$.
In this section, we discuss the two protocols that are obtained by replacing exactly one of the two kinds of conditions.
More precisely, one is the protocol called $\protocolProp \simpleXbasicFUFBD \property$, which consists of dialogues in $\protocolProp \UFBD \property$ satisfying
Conditions \ref{condition: distinct implies not equivalent simple UFBD},
\ref{condition: multiple elements X simple UFBD} and 
\ref{condition: last X simple UFBD} of Definition \ref{definition: simple UFBD} and Conditions \ref{condition2a: formula-based}, \ref{condition: no previous FB}, \ref{condition: non-empty FBs} and \ref{condition: empty FB is the last} of Definition \ref{definition: basic UFBD}; the other is the protocol called $\protocolProp \basicXsimpleFUFBD \property$, which consists of dialogues in $\protocolProp \UFBD \property$ satisfying 
Conditions
\ref{condition: un-poited-out property} and
\ref{condition: last X}
of Definition \ref{definition: basic UFBD}
and 
Conditions 
\ref{condition: pos or neg simple UFBD},
\ref{condition: FB prop should appear simple UFBD},
\ref{condition: FB prop should be new simple UFBD},
\ref{condition: no-unpointed property in simple UFBD} and
\ref{condition: last set F in simple UFBD}
of Definition \ref{definition: simple UFBD}.
The following diagram depicts the inclusion relation between the protocols.
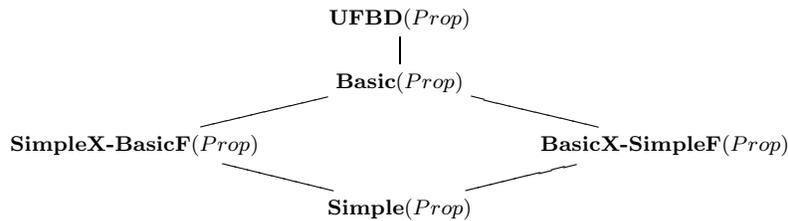
\[
\def\objectstyle{\scriptstyle}
\xymatrix@C=20pt@R=10pt{
                            &\protocolProp \UFBD \property&   \\
                          & \protocolProp \basicUFBD \property \ar@{-}[u]   &               \\
\protocolProp \simpleXbasicFUFBD \property  \ar@{-}[ru] &                                    &  \protocolProp \basicXsimpleFUFBD  \property \ar@{-}[lu] \\
             &    \protocolProp \simpleUFBD \property \ar@{-}[lu]\ar@{-}[ru]  &                 \\
}
\]

\noindent In what follows, we show that neither the protocol $\protocolProp \simpleXbasicFUFBD \property$ nor the protocol $\protocolProp \basicXsimpleFUFBD  \property$ necessarily has the convergence property.
This entails that the both kinds of the conditions in the definition of $\protocolProp \basicUFBD \property$ are required to be replaced to obtain the convergence property.

We first see an example to show that $\protocolTowards {\protocolSizeRestrict  {\protocolProp \simpleXbasicFUFBD {\propertyG O {\class C} }} n} \targetsetG$ does not necessarily converge to $\targetsetG$.
\begin{example}\label{example: SimpleX-BasicF PropG non-convergence}
We use the settings in Example \ref{example: target set in simple UFBD}.
Let $n$ be an integer that is strictly greater than $2$,  $\mc F'_1$ be the set \[\{([(\emptyset,\fami{E_{\syntax R}} {\syntax{R}} \relsetR)], pos), ([G_3], neutral), ([G_4], neutral)\}\] of FBs and $\mc F'_2$ be the empty set.
Then, we see:
\begin{enumerate}
\item
the sequence $\mc G_1, \mc F'_1, \mc G_1, \mc F'_2$ is a user-feedback dialogue in the protocol $\protocolTowards {\protocolSizeRestrict  {\protocolProp \simpleXbasicFUFBD {\propertyG O {\class C} }} n} \targetsetG$;

\item
there is no set $\mc G_2$ of FBs such that $\mc G_1, \mc F'_1, \mc G_1, \mc F'_2, \mc G_2$ is a dialogue in $\protocolTowards {\protocolSizeRestrict  {\protocolProp \simpleXbasicFUFBD {\propertyG O {\class C} }} n} \targetsetG$;

\item
the last set $\mc F'_2$ is not a non-empty subset of $\hypgraphset O {\class C}$.
\end{enumerate}
Therefore, the protocol  $\protocolTowards {\protocolSizeRestrict  {\protocolProp \simpleXbasicFUFBD {\propertyG O {\class C} }} n} \targetsetG$ does not converge to $\targetsetG$. \proved\checked \hfill $\dashv$
\end{example}

We next see that  $\protocolTowards {\protocolSizeRestrict  {\protocolProp \basicXsimpleFUFBD {\propertyG O {\class C} }} n} \targetsetG$ does not necessarily converge to $\targetsetG$.
\begin{example}
Again, we use the settings in Example \ref{example: target set in simple UFBD}.
Let $n$ be an integer strictly greater than $2$, $\mc G'_1$ be $\{G_3\}$ and $\mc F''_1$ be $\{([\emptyset, \fami {E_R}  R {\mc R})]_\simeq, pos)\}$.
Then, we see:
\begin{enumerate}
\item
the sequence $\mc G'_1, \mc F''_1, \mc G'_1$ is a user-feedback dialogue in the protocol $\protocolTowards {\protocolSizeRestrict  {\protocolProp \basicXsimpleFUFBD {\propertyG O {\class C} }} n} \targetsetG$;

\item
there is no set $\mc F''_2$ such that $\mc G'_1, \mc F''_1, \mc G'_1, \mc F''_2$ is a dialogue in the protocol $\protocolTowards {\protocolSizeRestrict  {\protocolProp \basicXsimpleFUFBD {\propertyG O {\class C} }} n} \targetsetG$;

\item
both of the following hold:
\begin{enumerate}
\item
$[( \{\syntax{p_3(0)} \}, \fami {E_R}  R \relsetR)]_{\simeq} \in \propertyG  O {\class C}(G)$ for any $G\in \mc G_t$ and 
\item
$[(\{\syntax{p_3(0)}\}, \fami {E_R}  R \relsetR )]_{\simeq} \notin \propertyG  O {\class C}(G_3)$.
\end{enumerate}
\end{enumerate}
Hence, the protocol $\protocolTowards {\protocolSizeRestrict  {\protocolProp \basicXsimpleFUFBD {\propertyG O {\class C} }} n} \targetsetG$ does not converge to $\targetsetG$.   \proved\checked \hfill $\dashv$
\end{example}

\section{Conclusion and Future Work}
\paragraph{Conclusion}
In this article we have defined the concepts of hypothesis and hypothesis graph that are allowed to contain many-sorted constants and variables as well as second-order predicates symbols whose arguments are first-order literals.
We have then proposed two types of user-feedback dialogue protocols, $\basicUFBD$ and $\simpleUFBD$, on hypotheses/hypothesis graphs, and shown that our protocols necessarily terminate under certain reasonable conditions and that they achieve hypotheses/hypothesis graphs that have the same properties in common as target hypotheses/hypothesis graphs do in common even if there are infinitely many hypotheses/hypothesis graphs.

\paragraph{Future Work}
To apply the proposed user-feedback dialogue protocols in real-world applications, we shall tackle the following two issues:
\begin{enumerate}
\item
Execution of any user-feedback dialogue protocol on hypothesis graph proposed in this article contains an extended version of subgraph isomorphism problem when checking whether a hypothesis graph satisfies a positive feedback.
Since subgraph isomorphism problem is NP-complete \cite{cook1997complexity}, it is required to approximate a user-feedback dialogue protocol when applying it to real world applications.
\item
To verify the usefulness of our protocols, we shall conduct an experiment to see whether our protocol terminates in a practical number of turns in real-world applications.
\end{enumerate}

\paragraph{Acknowledgements}
We would like to express our thanks the anonymous reviewers of the preliminary version \cite{motoura2023logic} for their valuable comments.

\section*{Declaration of Generative AI and AI-assisted technologies in the writing process}
\noindent Statement: During the preparation of this work the authors used DeepL and ChatGPT in order to improve English and sophisticate the names of some concepts in the article. After using this tool/service, the authors reviewed and edited the content as needed and takes full responsibility for the content of the publication.



\bibliographystyle{elsarticle-num} 
\bibliography{IJAR_arXiv.bib}





\end{document}